\documentclass[a4paper,UKenglish,cleveref, autoref, thm-restate]{lipics-v2021}

\title{On the hop-constrained Steiner tree problems}

\author{Adalat Jabrayilov}{
  Institute  of  Computer  Science,  University  of  Bonn,
  {Bonn, Germany}}
  {ajabrayilov@cs.uni-bonn.de}
  {https://orcid.org/0000-0002-1098-6358}{}
\authorrunning{Adalat Jabrayilov}
\Copyright{Adalat Jabrayilov}


\category{} %

\funding{This work was partially supported by DFG, RTG 1855.}%

\nolinenumbers %uncomment to disable line numbering

\EventEditors{John Q. Open and Joan R. Access}
\EventNoEds{2}
\EventLongTitle{42nd Conference on Very Important Topics (CVIT 2016)}
\EventShortTitle{CVIT 2016}
\EventAcronym{CVIT}
\EventYear{2016}
\EventDate{December 24--27, 2016}
\EventLocation{Little Whinging, United Kingdom}
\EventLogo{}
\SeriesVolume{42}
\ArticleNo{23}
\usepackage{graphicx}
\usepackage{amsmath} 
\usepackage{amssymb} 
\usepackage{todonotes} %
\usepackage[algo2e,ruled,vlined]{algorithm2e}
\usepackage{setspace}
\usepackage{numprint}
\usepackage{amsfonts} 
\usepackage{tikz} 
\usepackage{hyperref} 
\usepackage[utf8]{inputenc} %
  \usepackage{nicefrac} 
\usetikzlibrary[arrows,decorations.pathmorphing]

\begin{document} %

\maketitle

\setlength{\abovedisplayskip}{5pt}
\setlength{\belowdisplayskip}{5pt}
\def\NP{\textsf{NP}}

\def\subjectto{\textrm{s.t.}}
\def\Z{\mathbb{Z}}
\def\p{\mathcal{P}}
\def\Q{\mathcal{Q}}
\def\x{\mathcal{X}}
\def\X{\mathcal{X}}
\def\I{\mathcal I}
\def\APX{\mathcal{APX}}
\def\W{\mathcal{W}}
\def\V{\mathbb{V}}
\def\po{partial-ordering}
\def\pop{\textrm{POP}}
\def\fme{Fourier-Motzkin elimination\xspace}

\def\hdmcf{\textrm{HMCF}\xspace}

\def\head{\text{head}}
\def\tail{\text{tail}}

\def\ie{i.e.}
\def\eg{e.g.}
\def\Wlog{W.l.o.g.}
\def\len#1{\ell_{#1}}
\def\inarc{\delta^{-}}
\def\red#1{\textcolor{red}{#1}}

\def\todos{0}
\def\todoI#1{
    \ifnum\todos=1 %
      \todo[inline]{#1}
    \fi
  }
\def\vs{\vspace{1cm}}

\def\mtodo#1{\todo{\tiny #1} }

\def\S{\text S}
\def\J{\text F}
\def\JJ{\mathcal F}
\def\walk{\text {W}}
\def\Jround#1{J^{(#1)}}
\def\Jused{J^{\emph{used}}}
\def\Y{Y}
\def\yvar{$y$ variable}
\def\yvars{$y$ variables}

\def\setR{\mathbb{R}\xspace}
\def\setN{\mathbb{N}\xspace}
\def\pos{\pi}
\def\minimize{\textrm{min}}
\def\hs{\hspace{8pt}} 
\def\bc{branch-and-cut\xspace}
\def\BC{\emph{branch-and-cut}\xspace}
\def\bp{\emph{branch-and-price}\xspace}
\def\SL{Sinnl-Ljubi\'c\xspace}
\def\Ljubic{Ljubi\'c\xspace}
\def\alert#1{\textcolor{red}{#1}}

\def\stprbh{\textrm{STPRBH}\xspace}
\def\hstp{\textrm{HSTP}\xspace}
\def\hmst{\textrm{HMSTP}\xspace}
\def\stp{\textrm{STP}\xspace}
\def\mst{\textrm{MSTP}\xspace}
\def\ppop{\mathcal{P}}
\def\pass{\mathcal{A}}

\def\ppopHSTP{\ppop^{\hstp}}
\def\ppopHMST{\ppop^{\hmst}}
\def\ppopSTPRBH{\ppop^{\stprbh}}

\def\xpopHSTP{\ppopHSTP_X}
\def\xpopHMST{\ppopHMST_X}
\def\xpopSTPRBH{\ppopSTPRBH_X}

\def\passHSTP{\pass^{\hstp}}
\def\passHMST{\pass^{\hmst}}
\def\passSTPRBH{\pass^{\stprbh}}

\def\xassHSTP{\passHSTP_X}
\def\xassHMST{\passHMST_X}
\def\xassSTPRBH{\passSTPRBH_X}

\def\passHT{\pass}    %
\def\passHTstr{\pass^*}
\def\xassHTstr{\pass^*_X}

\def\ppopHT{\ppop}	%

\def\lpassHT{\textrm{(A-HT)}\xspace}
\def\lppopHT{\textrm{(P-HT)}\xspace}

\def\xassHT{\passHT_X}
\def\xpopHT{\ppopHT_X}

\def\passtopop{\pass_{\ppop}}
\def\ppoptoass{\ppop_{\pass}}
\def\passtopopHSTP{\pass_{\ppop}^{\hstp}}
\def\ppoptoassHSTP{\ppop_{\pass}^{\hstp}}

\def\val{\nu}
\def\gap#1{gap_{#1}}

\def\msum{\textstyle\sum}

\newif\ifshort	%
\shorttrue
\shortfalse

\begin{abstract}
The \emph{hop-constrained Steiner tree problem (\hstp)} is a generalization of the classical \emph{Steiner tree problem}. It asks for a minimum cost subtree that spans some specified nodes of a given graph, such that the number of edges between each node of the tree and its root respects a given hop limit.  
This \NP-hard problem has many variants, often modeled as integer linear programs.
Two of the models are so-called \emph{assignment} and \emph{partial-ordering} based models, which yield (up to our knowledge) the best two state-of-the-art formulations for the variant \emph{Steiner tree problem with revenues, budgets, and hop constraints~(\stprbh)}.
The solution of the \hstp and its variants such as the \stprbh and the \emph{hop-constrained minimum spanning tree problem (\hmst)} is a hop-constrained tree, a rooted tree whose depth is bounded by a given hop limit.
This paper provides some theoretical results that show the polyhedral advantages of the partial-ordering model over the assignment model for solving this class of problems.
Computational results in this paper and the literature for the \hstp, \stprbh, and \hmst show that the partial-ordering model outperforms the assignment model in practice, too; it has better linear programming relaxation and solves more instances. 
\end{abstract}

\section{Introduction}
\ifshort
The \emph{Steiner tree problem (STP)} belongs to the classical \NP-hard~\cite{doi:10.1137/0132072} optimization problems. 
Given a weighted undirected graph $G=(V,E)$ with edge costs $c \colon E \rightarrow \setR^{>0}$ and a subset $R \subseteq V$ of nodes, called \emph{terminals}, the STP asks for a subtree $T$ of $G$ that contains all terminals and has minimum costs~$\sum_{e \in E(T)} c_e $.
\else
Many network design applications ask for a minimum cost subtree connecting some required nodes of a graph.
These applications can be modeled as the \emph{Steiner tree problem (STP)}:
Given a weighted undirected graph $G=(V,E)$ with positive edge costs $c \colon E \rightarrow \setR^{>0}$ and a subset $R \subseteq V$ of the required nodes, called \emph{terminals}, find a subtree $T$ of $G$, which contains all terminals and has minimum costs, \ie,~$\sum_{e \in E(T)} c_e $ is minimal.
The STP belongs to the classical optimization problems and is \NP-hard (Garey et al.~\cite{doi:10.1137/0132072}).
\fi
For a survey on the STP, see \cite{HwangRichards92,winter1987steiner}.
\ifshort
In the special case $R=V$, the problem is known as the \emph{minimum spanning tree problem} (\mst) and is solvable in polynomial time~\cite{Kruskal1956,Prim1957}.
\else
In the special case $R=V$, the problem is known as the \emph{minimum spanning tree problem} (\mst) and is solvable in polynomial time (Kruskal \cite{Kruskal1956}, Prim \cite{Prim1957}).
\fi
\ifshort
Hop constraints, also known as \emph{height constraints} \cite{ManyemStallman96}, are originated from telecommunication applications and limit the number of hops (edges) between the given service provider (root) and terminals of the network to control the quality of the service \cite{woolston1988design}. 
\else
Hop constraints, also known as \emph{height constraints} \cite{ManyemStallman96}, are originated from telecommunication applications and limit the number of hops (edges) between the given service provider (root) and terminals of the network to control the availability and reliability of the service (Woolston and Albin \cite{woolston1988design}). 
\emph{Availability} is the probability that all edges in the connection (path) between the server and a terminal are working.  
\emph{Reliability} is the probability that this connection will not be interrupted by an edge failure. 
The failure probability of the path with at most $H$ edges does not exceed $1-(1-p)^H$, where $p$ is the failure probability of an edge.
\fi
For a survey on more general network design problems with hop constraints, see \cite{doi:10.1287/ijoc.4.2.192,kerivin2005design,DBLP:series/natosec/MonteiroFF14}.
\ifshort
\else

\fi
\emph{Hop-constrained Steiner tree problem} (\hstp) %
is defined as follows: 
Given a weighted undirected graph $G=(V,E)$ with 
edge costs $c \colon E \rightarrow \setR^{>0}$,
a set $R \subseteq V$ of terminals, 
a root $r \in R$,
and a hop limit $H \in \setN^{\ge0}$,
find a minimum cost subtree $T$ of the graph, 
which contains all terminals such that
the number of edges in the path from $r$ 
to any $v \in V(T)$ does not exceed %
$H$. 
\ifshort
\else
There are several variants of the \hstp.
\fi
The well-studied special case $R=V$ is known as the \emph{hop-constrained minimum spanning tree problem} (\hmst).
\ifshort
While the \mst is solvable in polynomial time, the \hmst is \NP-hard~\cite{Dahl98the2hopMST,GOUVEIA1995959}.
\else
While the \mst is solvable in polynomial time, the \hmst is \NP-hard (Dahl \cite{Dahl98the2hopMST}, Gouveia \cite{GOUVEIA1995959}).
Moreover, the \hmst\ is not in $\APX$ (Manyem and Stallmann \cite{ManyemStallman96}), 
\ie, the class of problems for which it is possible to have a polynomial-time constant-factor approximation. 
\fi
\emph{Steiner tree problem with revenues, budgets, and hop constraints} (\stprbh)
is another well-studied variant of the \hstp: 
Given a weighted undirected graph $G=(V,E)$
\ifshort
with edge costs $c \colon E \rightarrow \setR^{>0}$,
node revenues $\rho \colon V \rightarrow \setR^{\ge0}$,
\else
with positive edge costs $c \colon E \rightarrow \setR^{>0}$,
nonnegative node revenues $\rho \colon V \rightarrow \setR^{\ge0}$,
\fi
a root $r \in V$,
a hop limit $H \in \setN^{\ge0}$,
and a budget $B \in \setR^{\ge0}$, %
find a subtree $T$ %
of the graph,
which contains $r$, 
maximizes the collected revenues $\sum_{v \in V(T)} \rho_v$,
such that
the number of edges in the path from $r$ to any node $v \in V(T)$ does not
exceed %
$H$, 
and the total edge
costs of the tree respect the budget $B$, \ie\ $\sum_{e \in E(T)} c_e \le B$.

The majority of the integer linear programming (ILP) formulations in the literature for these \NP-hard problems can be divided into \emph{``node-oriented''} and \emph{``edge-oriented''} formulations.
To formulate the hop constraints, the node-oriented models use node variables, which describe the depth of nodes in the tree.     
Akgün and Tansel \cite{Akg_n_2011}, Gouveia \cite{GOUVEIA1995959} have presented several node-oriented formulations based on Miller-Tucker-Zemlin (MTZ) subtour elimination constraints \cite{MillerTuckerZemlin60} for the \hmst.
The MTZ constraints involve an integer variable 
\ifshort
$U_v$
\else
$U_v \in \{1,\dots,H\}$ 
\fi
for each node $v$, which specifies the depth of $v$ in the solution tree, \ie, the number of edges between the root and $v$.
Voß \cite{DBLP:journals/anor/Voss99} has used the MTZ constraints to model the \hstp.
Costa et al.~\cite{Costa2009ModelsAB} and Layeb et al.~\cite{Layeb2013SolvingTS} have proposed several models for the \stprbh\ based on the MTZ constraints. 
Sinnl and \Ljubic \cite{DBLP:journals/mpc/SinnlL16} have suggested an assignment model for the \stprbh, which uses a binary variable $y_{vh}$ for each node $v$ and each depth $h$, where $y_{vh}=1$ indicates that the depth of $v$ in the tree is $h$.
\ifshort
Recently, Jabrayilov and Mutzel \cite{DBLP:conf/alenex/JabrayilovM19} have presented a partial-ordering based model for the \stprbh,
where the variables indicate whether the depth of $v$ is greater or less than $h$.
\else
Recently, Jabrayilov and Mutzel \cite{DBLP:conf/alenex/JabrayilovM19} have presented a partial-ordering based model for the \stprbh.
Instead of directly assigning a depth $h$ to node $v$, the variables of this model indicate whether the depth of $v$ is greater or less than $h$.
\fi

The edge-oriented formulations use edge variables to describe the hop constraints.
Gouveia \cite{GouveiaMCF} has presented a multicommodity flow (MCF) model for the \hmst. 
Although this edge-oriented model's linear programming~(LP) bound is much better than the MTZ model's~\cite{GOUVEIA1995959}, it leads to large ILP models.
While the MTZ model has $O(|V|^2)$ variables and constraints, the MCF model has $O(|V|^3)$ variables and constraints.
Gouveia \cite{GouveiaRedef} has proposed the %
\emph{``hop-dependent''} multicommodity flow (\hdmcf) model for the \hmst\ and \hstp, yielding better LP bounds than the MCF model. 
The \hdmcf, with its $O(H|V|^3)$ variables, is even larger than the MCF and applicable to small graphs.
Gouveia et al.~\cite{GOUVEIA2001539} have reported that the \hdmcf model cannot solve the LP relaxation after a couple of days for most instances with 40 nodes.
They have introduced a Lagrangian relaxation for the \hdmcf model.
Gouveia et al.~\cite{Gouveia2011ModelingHA} modeled the \hmst as the STP in the so-called layered graph and introduced an ILP with $O(H|E|)$ variables and an exponential number of constraints for this STP.  
Costa et al.~\cite{Costa2009ModelsAB} have presented an edge-oriented model with $O(H|E|)$ variables for the \stprbh. %
There are also some heuristics (Costa et al.~\cite{COSTA200868}, Fernandes et al.~\cite{fernandes2007determining}, Fu and Hao \cite{DBLP:journals/eor/FuH14,2015:DPD:2748889.2748890}, Gouveia et al.~\cite{Gouveia2011RestrictedDP}), and ILP models with an exponential number of constraints or variables, which require sophisticated \BC or \bp algorithms (Costa et al.~\cite{Costa2009ModelsAB}, Dahl et al.~\cite{DBLP:books/daglib/p/DahlGR06}, Sinnl~\cite{sinnl11}).

The advantage of node- and edge-oriented formulations is that they have a polynomial size and can be fed directly into a standard ILP solver.
Although the latter have stronger LP relaxations than the former, they lead to large ILPs and are suitable for small graphs. 
On the other hand, the node-oriented models use far fewer variables and allow to tackle large instances.
In fact, the best state-of-the-art models for the %
\stprbh are node-oriented models, namely the assignment and the partial-ordering based models. %
Computational results in the literature show that for a majority of the large \stprbh instances with 500 nodes and \numprint{12500} edges from the DIMACS benchmark set~\cite{DIMACS11BENCHMARK}, these two models find optimal integer solutions within seconds~\cite{DBLP:conf/alenex/JabrayilovM19,DBLP:journals/mpc/SinnlL16}, whereas the previous models cannot even solve the LP relaxation within a time limit of two hours~\cite{COSTA200868}.

\noindent\textit{Our contribution.}
The solution of the \hstp, \hmst and \stprbh is a hop-constrained tree, a rooted tree whose depth is bounded by a given hop limit.
In this paper, we provide some theoretical results that show the polyhedral advantages of the partial-ordering model over the assignment model for 
solving this class of problems.
Computational results in this paper and the literature~\cite{DBLP:conf/alenex/JabrayilovM19} for the \hstp, \hmst and \stprbh show that the partial-ordering model outperforms the assignment model in practice, too. 
It has better linear programming relaxation and solves more instances.

\noindent\textit{Outline.}
We start with some notations %
(\autoref{sec:notations}). 
In \autoref{sec:descripton:asspop}, we describe the assignment and partial-ordering based models. 
The polyhedral and computational results are presented in Sections \ref{sec:assvspop} and \ref{sec:evaluation}. %
We conclude with \autoref{sec:conclusion}.

\section{Notations} 
\label{sec:notations}
For a graph $G=(V,E)$, we denote its node set by $V(G)$ and edge set by $E(G)$.
Each edge of an undirected graph $G$ is a 2-element subset $e=\{u,v\}$ of $V$.
For clarity, we may write $e=uv$.
\ifshort\else
The end nodes of an edge are called neighbors.
\fi
Each edge of a directed graph is an ordered pair $a=(u,v)$ of nodes and is called an arc. 
The arc $(u,v)$ is outgoing from $u$ and incoming to $v$.
\ifshort\else
The node $u$ is the tail of $a$, written $u=\tail(a)$, and $v$ is the head of $a$, written $v=\head(a)$.
\fi
A directed walk $W$ in a graph is a sequence $W=v_0,a_1,v_1,\dots,v_{k-1},a_{k},v_k$ of nodes and arcs, such that $a_i=(v_{i-1},v_i)$ for $1\le i \le k$.
We may describe $W$ as the sequence of its nodes $v_0,\ldots,v_{k}$ or arcs $a_1,\ldots,a_k$. 
We call $W$ a $(v_0,v_k)$-walk.
Let $A$ be a list or set of arcs $a_1,\dots,a_k$, and let $x_a$ be a variable associated with arc $a$.
We use the abbreviation $x(A):=x_{a_1}+\dots+x_{a_k}$.
A walk is a path if all its nodes are distinct. 
The number of not necessarily distinct arcs in $W$ is called the length of $W$, written $\len W$.
We denote the set of incoming arcs to a node $v$ in $W$ by $\inarc_W(v)$.
\ifshort\else
We say that $W'$ is a suffix of walk $W=a_1,\dots,a_l$, written $W' \sqsupset W$, if there is an $1\le i \le l$, such that $W'=a_i,\dots,a_l$.
\fi
A rooted tree $T$ has a special node $r \in V(T)$.
The depth $d(v)$ of node $v$ in $T$ is the number of edges in the $(r,v)$-path in $T$.
The depth of $T$ is $d(T):=\max \{d(v) \colon v \in V(T)\}$.
We say $T$ contains an arc $(u,v)$, written $(u,v) \in E(T)$, if $T$ contains an edge $uv$, such that $d(u)<d(v)$.
\ifshort
For a given $H$ we call a rooted tree $T$ with $d(T) \le H$ a \emph{hop-constrained tree} (HT).
\else
For a given hop limit $H \in \setN^{\ge0}$, we call a rooted tree $T$ with $d(T) \le H$ a \emph{hop-constrained tree} (HT).
\fi
\ifshort
\else
Let $C$ be an inequality in the form $L \le R$ with a linear expression $L$ and a number $R$. 
By $c(var,C)$, we denote the coefficient of variable $var$ in $C$, \eg, $c(y, x-4y \le 7)=-4$.
\fi
We denote by $\val_M$ the LP relaxation value of ILP $M$.
Let $A$ and $B$ be two ILPs for problem~$P$.
In case $P$ is a minimization (resp. maximization) problem, we call $A$ stronger than $B$, if inequality $\val_A \ge \val_B$ (resp. $\val_A \le \val_B$) holds for all $P$ instances; $A$ is strictly stronger than $B$ if $A$ is stronger than $B$, but $B$ is not stronger than $A$, \ie, there is a $P$ instance for which the inequality is strict. 
\section{Assignment and partial-ordering based models} %
\label{sec:descripton:asspop}
The main idea of both models is similar: The solution of the hop-constrained Steiner tree problems is a hop-constrained subtree $T$ of $G$ with root $r$ and depth $d(T) \le H$.
To describe the arcs of $T$, both models use two binary variables $x_{u,v}$, $x_{v,u}$ for each edge $uv \in E$: $x_{u,v}=1$ if and only if $(u,v) \in E(T)$, and $x_{v,u}=1$ if and only if $(v,u) \in E(T)$.
Both models compute a partial order $\pi$ of nodes of $T$ such that 
the position $\pi(v)$ of any node $v\in V(T)$ in this order
satisfies $d(v) \le \pi(v) \le H$, 
and thus 
$d(v)$ respects the hop-limit $H$.
To describe the positions of nodes in $\pi$, 
the assignment model 
\cite{DBLP:journals/mpc/SinnlL16}
for the \stprbh\ 
uses a binary variable $y_{v,i}$ for each node $v \in V$ and each position 
$i \in \{0,\dots,H\}$; 
$y_{v,i}=1$ if and only if $\pi(v)=i$.
The partial-ordering-based model \cite{DBLP:conf/alenex/JabrayilovM19} uses instead of the variables $y$, so-called \emph{partial-ordering problem (POP)} variables, 
namely binary variables $l$ (less than) and $g$ (greater than).
For each node $v \in V$ and each position $i$, %
$l_{v,i}=1$ if and only if $\pi(v) < i$,
and 
$g_{i,v}=1$ if and only if $\pi(v) > i$.
That is, if node $v$ is at position $i$, then we have 
$l_{v,i}=g_{i,v}=0$, \ie, 
the connection between the assignment and POP variables is: 
\begin{align}
    y_{v,i} = 1- (l_{v,i}+g_{i,v}) && 
    & \forall v \in V,\ i = 0,\dots,H.
    \label{connection:asspop}
\end{align}
We may assume that $G$ is complete; otherwise, we can assign infinite costs to the missing edges.
\ifshort
The partial-ordering based model \lppopHT for hop-constrained trees (HT) can be formulated similar to the \stprbh model \cite{DBLP:conf/alenex/JabrayilovM19}: 
\else
Via $x,l,g$ variables, the partial-ordering based model \lppopHT for hop-constrained trees (HT) can be formulated similar to the \stprbh model \cite{DBLP:conf/alenex/JabrayilovM19}: 
\fi
\begin{align}
    \lppopHT
    \ \ \ \ \ 
    \ \ \ \ \ 
    & l_{r,0} = g_{0,r} = 0
    \label{constr:root}
    \\
    & l_{v,1} = g_{H,v} = 0
    & \forall v \in V \setminus \{r\}
    \label{constr:interval}
    \\
    & l_{v,i} - l_{v,i+1} \le 0
    & \forall v \in V,\ i = 0,\dots,H-1
    \label{constr:unambiguous1}
    \\
    & g_{i,v} + l_{v,i+1} = 1 
    & \forall v \in V,\ i = 0,\dots,H-1
    \label{constr:unambiguous2}
    \\
    & l_{u,i} + g_{i,v} \ge x_{u,v}
    & \forall u \in V, v \in V\setminus \{u\}, i = 0,\dots,H
    \label{constr:edge:forward} 
    \\
    & \msum_{u \in V \setminus \{v\}} x_{u,v}  \le 1
    & \forall v \in V
    \label{constr:indeg} 
    \\
    & \msum_{u \in V \setminus \{v,w\}} x_{u,v}  \ge  x_{v,w}
    & \forall v \in V \setminus \{r\},\ w \in V \setminus \{v\}
    \label{constr:outdeg} 
    \\
    & x \in [0,1]^{2|E|}  
    \label{constr:xvars} 
    \\
    & l,g \in [0,1]^{H|V|}  \label{pop:lgvars}
\end{align}
Constraints (\ref{constr:root}) ensure that $\pi(r)=0$,
and (\ref{constr:interval}) 
make sure that $\pi(v) \in \{1,\dots,H\}$ 
for each %
$v \in V \setminus \{r\}$.
If a node's position is less than $i$, it is also less than $i+1$ by (\ref{constr:unambiguous1}).
Constraints (\ref{constr:unambiguous2}) ensure that for each node $v$, either $\pi(v)>i$   (\ie\ $g_{i,v}=1$) or $\pi(v)<i+1$ (\ie\ $l_{v,i+1}=1$), and not both.
The constraints 
(\ref{constr:unambiguous1})
jointly with (\ref{constr:unambiguous2})
enforce each node $v$ to have exactly one position in $\pi$ \cite{DBLP:conf/alenex/JabrayilovM19}.
Recall that since $G$ is complete, there is a variable $x_{uv}$ for each two distinct nodes $u \ne v \in V$.
Constraints (\ref{constr:edge:forward}) ensure $\pi(u) < \pi(v)$ for each arc $ (u,v) \in E(T)$ \cite{DBLP:conf/alenex/JabrayilovM19}.
Each node has at most one incoming arc in $T$, which is ensured by (\ref{constr:indeg}).
Constraints (\ref{constr:outdeg}) express that each node $v \in V \setminus \{r\}$ with an outgoing arc has also an incoming arc. %
The fact that $\pi(u) < \pi(v)$ for each arc $(u,v) \in E(T)$, jointly with 
Constraints (\ref{constr:indeg})--(\ref{constr:xvars}), 
ensure that $T$ is a tree \cite{DBLP:conf/alenex/JabrayilovM19}.  
Moreover,
$d(T) \le H$ by (\ref{constr:root})--(\ref{constr:edge:forward}), %
\ie, $T$ is an HT. %
We denote the polytope of %
\lppopHT by $\ppopHT$:
\begin{gather} \label{ppop}
 \ppopHT := \{ (x,l,g) \colon (x,l,g) \mbox{ satisfy } 
 (\ref{constr:root})-(\ref{pop:lgvars}) \}.
\end{gather}
Using the variables $x,y$, the assignment model %
\lpassHT
for hop-constrained trees (HT)
can be formulated similar to the \stprbh model \cite{DBLP:journals/mpc/SinnlL16}:
\begin{align}
    \lpassHT
    \ \ \ \ 
    & \mbox{Constraints (\ref{constr:indeg})--(\ref{constr:xvars}) } \notag
    \\
    & y_{r,0} = 1
        \label{ass:constr:root}
    \\
    & y_{r,i} = 0
    & \forall i = 1,\dots,H
        \label{ass:constr:root:i}
    \\
    & y_{v,0} = 0
    &	\forall v \in V \setminus \{r\}
        \label{ass:constr:Vroot}
    \\
    & \msum_{i=1}^H y_{v,i} = 1
    &	\forall v \in V \setminus \{r\}
        \label{ass:constr:interval}
    \\
    & y_{u,i} - y_{v,i+1} + x_{u,v} \le 1
    & \forall u \in V, v \in V\setminus \{u\}, i = 0,\dots,H-1
    \label{ass:constr:arcdir} 
    \\
    & y_{u,H} + x_{u,v} \le 1
    & \forall u \in V, v \in V\setminus \{u\}
    \label{ass:constr:arcdirH} 
    \\
    & y \in [0,1]^{H|V|} 
    \label{ass:vars} 
\end{align}
As both models use the same $x$-variables, \lpassHT also includes the constraints (\ref{constr:indeg})--(\ref{constr:xvars}). %
Constraints (\ref{ass:constr:root}) 
and (\ref{ass:constr:root:i}) 
ensure $\pi(r)=0$. %
Constraints (\ref{ass:constr:Vroot}) and (\ref{ass:constr:interval}) 
make sure that $\pi(v) \in \{1,\dots,H\}$ for each $v \in V \setminus \{r\}$.
Constraints (\ref{ass:constr:arcdir})
express that if $x_{u,v}=1$, then $\pi(v)=\pi(u)+1$.
Constraints (\ref{ass:constr:arcdirH}) make sure that if $\pi(u)=H$, then $u$ has no outgoing arc $(u,v)$ in $T$.
From (\ref{ass:constr:arcdir}) and (\ref{ass:constr:arcdirH}) follows that for each arc $(u,v) \in T$, we have $\pi(u) < \pi(v)$.
Like in model \lppopHT, 
the fact that $\pi(u) < \pi(v)$ for each arc $(u,v) \in E(T)$, jointly with 
(\ref{constr:indeg})--(\ref{constr:xvars}), 
guarantee that $T$ is a tree. 
Moreover,
$d(T) \le H$ by (\ref{ass:constr:root})--(\ref{ass:vars}),
\ie, $T$ is an HT. %
We denote the polytope of the assignment model \lpassHT by $\passHT$:
\begin{gather} 
\passHT := \{ (x,y) \colon (x,y) \mbox{ satisfy } 
(\ref{constr:indeg})-(\ref{constr:xvars}),
(\ref{ass:constr:root})-(\ref{ass:vars})\}
\label{pass1} %
\end{gather}

\noindent
The solution $T$ of the \hstp is a HT that
contains
all terminals $R \subseteq V$, \ie:
\begin{align}
  &\msum_{u \in V \setminus \{v\}} x_{u,v}  \ge 1
  & \forall v \in R \setminus \{r\}
  \label{constr:indeg:ge1} 
\end{align}
Let $\passHSTP$ and $\ppopHSTP$ 
denote the polytopes of the assignment and
partial-ordering based LPs for the \hstp, respectively. Then we have:
\begin{align*}
\passHSTP = & \{ (x,y) \colon (x,y) \in \passHT\ (\ref{pass1}),\ x \mbox{ satisfies } (\ref{constr:indeg:ge1})\}, \\
\ppopHSTP = & \{ (x,l,g) \colon (x,l,g) \in \ppopHT\ (\ref{ppop}),\ x \mbox{ satisfies } (\ref{constr:indeg:ge1})\}.
\end{align*}
The corresponding ILPs have the following form:
\ifshort
\begin{align}
  \min & \big\{ 
    \msum_{uv \in E} c_{u,v} ( x_{u,v} + x_{v,u} )
    \colon
    (x,y) \in \passHSTP;\  x,y \text{ integral}
    \big\},
    \label{Ahstp}
    \tag{A-\hstp}
    \\
  \min & \big\{ 
    \msum_{uv \in E} c_{u,v} ( x_{u,v} + x_{v,u} )
    \colon
    (x,l,g) \in \ppopHSTP;\  x,l,g \text{ integral}
    \big\}.
    \label{Phstp}
    \tag{P-\hstp}
\end{align}
\else
\begin{align}
  \min & \Big\{ 
    \msum_{uv \in E} c_{u,v} ( x_{u,v} + x_{v,u} )
    \colon
    (x,y) \in \passHSTP;\  x,y \text{ integral}
    \Big\},
    \label{Ahstp}
    \tag{A-\hstp}
    \\
  \min & \Big\{ 
    \msum_{uv \in E} c_{u,v} ( x_{u,v} + x_{v,u} )
    \colon
    (x,l,g) \in \ppopHSTP;\  x,l,g \text{ integral}
    \Big\}.
    \label{Phstp}
    \tag{P-\hstp}
\end{align}
\fi
\section{Polyhedral comparison} %
\label{sec:assvspop}

\subsection{Direct comparison of polytopes $\passHSTP$ and $\ppopHSTP$} 
\ifshort
To compare (\ref{Ahstp}) and (\ref{Phstp}), we can directly transform their polytopes onto the same variable space via (\ref{connection:asspop}), \ie,
\begin{align*}
\passHSTP \mbox{ to } \passtopopHSTP
:=\{ (x,l,g) \colon (x,y) \in \passHSTP \text{ and } y,l,g \mbox{ satisfy } (\ref{connection:asspop}) \},
\\
\ppopHSTP \mbox{ to } \ppoptoassHSTP
:=\{ (x,y) \colon (x,l,g) \in \ppopHSTP \text{ and } y,l,g \mbox{ satisfy }  (\ref{connection:asspop}) \},
\end{align*}
and compare them. However, the examples below show that 
$ \ppoptoassHSTP \setminus \passHSTP \ne \emptyset $ and
$\passtopopHSTP \setminus \ppopHSTP \ne \emptyset$, so %
$\passHSTP$ and $\ppopHSTP$ are not comparable.
\else
To compare (\ref{Ahstp}) and (\ref{Phstp}) directly, we can replace, for example, in (\ref{Ahstp}) $y$ variables by $l,g$ via (\ref{connection:asspop}), and then test if the resulted constraints imply~(``$\Rightarrow$'') constraints of (\ref{Phstp}) or vice versa~(``$\Leftarrow$'').
However, the examples below that handle the transformations of polytopes of (\ref{Ahstp}) and (\ref{Phstp}) onto the same variable space, namely
\begin{align*}
\passHSTP \mbox{ to } \passtopopHSTP
:=\{ (x,l,g) \colon (x,y) \in \passHSTP \text{ and } y,l,g \mbox{ satisfy } (\ref{connection:asspop}) \},
\\
\ppopHSTP \mbox{ to } \ppoptoassHSTP
:=\{ (x,y) \colon (x,l,g) \in \ppopHSTP \text{ and } y,l,g \mbox{ satisfy }  (\ref{connection:asspop}) \},
\end{align*}
show that
$ \ppoptoassHSTP \setminus \passHSTP \ne \emptyset $ and
$\passtopopHSTP \setminus \ppopHSTP \ne \emptyset$, so 
neither ``$\Rightarrow$'' nor ``$\Leftarrow$'' holds.
\fi

\begin{figure}[bt]
  \centering
  \captionsetup[subfigure]{justification=centering}
  \begin{subfigure}[c]{0.34\textwidth}
	\scalebox{0.8}{
	  \begin{tikzpicture}[
	    x=1.1 cm, y=1.1 cm,
	    V/.style= {circle, inner sep=4pt, draw, line width=1pt},
	    P/.style= {rectangle, inner sep=1pt, draw},
	    E/.style= {->,draw, line width=3},
	    L/.style  = {bend left=13pt},
	    R/.style  = {bend right=13pt},
	    Pointer/.style= {->, draw, opacity=0.5},
	  ]
	    \path [draw, opacity=0.1, step=1.1cm] (-0.2,0.6) grid (2.8,3.3);

	    \path (2.7,2.7) node [] {{\LARGE $G$}};
	    \path (1.4,2.3) node [] {$T$};

	    \path (2,3)	node [V] (r) {$r$};
	    \path (1,1)   node [V] (a) {$a$};
	    \path (2,1)	node [V] (b) {$b$};
	    \path (3,1)	node [V] (c) {$c$};

	    \path (r) edge [E] (a);
	    \path (r) edge [E] (b);
	    \path (r) edge [E] (c);

	    \path (a) edge [L] (b);
	    \path (b) edge [L] (c);
	    \path (a) edge [bend right=32] (c);

	    \path (0.0,3.2) edge [Pointer, line width=2pt, opacity=0.1] +(0,-2.6) 
	    +(-0.3,-2.5) node {$\pos$};

	    \path (0, 3) node [P] (p1) {} +(0.3,0) node {$0$};
	    \path (0, 2) node [P] (p2) {} +(0.3,0) node {$1$}; 

	    \path (0, 1) node [P] (p3) {} +(0.3,0) node {$2$}; 
	  \end{tikzpicture}
	}
	\subcaption{}
	\label{fig:pop}
  \end{subfigure}
  \hspace{32pt}
  \begin{subfigure}[c] {0.32\textwidth}
	\scalebox{0.8}{
	  \begin{tikzpicture}[
	    x=1.1 cm, y=1.1 cm,
	    V/.style= {circle, inner sep=4pt, draw, line width=1pt},
	    P/.style= {rectangle, inner sep=1pt, draw},
	    E/.style= {->, line width=3},
	    L/.style  = {bend left=11pt},
	    R/.style  = {bend right=11pt},
	    Pointer/.style= {->, draw, opacity=0.5},
	  ]

	    \path (2.9,2.7) node [] {{\LARGE $G$}};
	    \path (1.0,1.7) node [] {$T$};

	    \path (2,3)		node [V] (r) {$r$};
	    \path (0,1.3)	node [V] (a) {$a$};
	    \path (2,1.7)	node [V] (b) {$b$};
	    \path (4,1.3)	node [V] (c) {$c$};

	    \path (r) edge [] (a);
	    \path (r) edge [] (b);
	    \path (r) edge [] (c);

	    \path (a) edge [E] node [below] {$\frac23$}  (b);
	    \path (b) edge [E] node [below] {$\frac23$}  (c);
	    \path (c) edge [bend left=27pt, E] node [above] {$\frac23$}  (a);
	  \end{tikzpicture}
	}
	\subcaption{}
	\label{fig:ass}
  \end{subfigure}
  \ifshort
  \caption{Graph $G$ and feasible solutions (\subref{fig:pop}) for (\ref{Phstp}) and (\subref{fig:ass}) for (\ref{Ahstp})}
  \else
  \caption{Complete graph $G$ and feasible solutions (\subref{fig:pop}) for (\ref{Phstp}) and (\subref{fig:ass}) for (\ref{Ahstp})}
  \fi
  \label{fig:example}
\end{figure}

\begin{example} [$\ppoptoassHSTP \setminus \passHSTP \ne \emptyset$] \label{example:pop-ass}
Let $G=(V,E)$ be a complete graph, $V=R=\{r,a,b,c\}$, and $H=2$.
Let $\pi^{\pass}$ and $\pi^{\ppop}$ be the partial orders constructed 
by %
(\ref{Ahstp}) and (\ref{Phstp}), respectively. 
Recall that for an arc $(u,v) \in T$,
$\pi^{\pass}$ satisfies
$\pi^{\pass}(v)=\pi^{\pass}(u)+1$,  %
while $\pi^{\ppop}$ meets
$\pi^{\ppop}(v) \ge \pi^{\ppop}(u)+1$, 
\eg, $\pi^{\ppop}(v) = \pi^{\ppop}(u)+2$ is possible, too.
Using this observation, we construct feasible solution $(x,l,g) \in \ppopHSTP$ as in Figure \ref{fig:pop}:
The $x$ values are all zero, except $x_{r,a}=x_{r,b}=x_{r,c}=1$. 
Moreover, $\pi^{\ppop}(r)=0$
and $\pi^{\ppop}(a)=\pi^{\ppop}(b)=\pi^{\ppop}(c)=2$, \ie,
$(g_{0,r},g_{1,r},g_{2,r})=(0,0,0)$,
and
$(g_{0,v},g_{1,v},g_{2,v})=(1,1,0)$
for $v \in \{a,b,c\}$.
For clarity, %
we do not write $l$ values, %
since they can be derived from %
(\ref{constr:unambiguous2}).
Via (\ref{connection:asspop}), we 
transform $(x,l,g)\in \ppopHSTP$ to $(x,y)\in\ppoptoassHSTP$,
which violates (\ref{ass:constr:arcdir}):
$y_{r,0} - y_{a,1} + x_{r,a} = 
(1- (l_{r,0}+g_{0,r})) - (1- (l_{a,1}+g_{1,a})) + x_{r,a} =
(1-(0+0)) - (1-(0 + 1)) + 1 =
2 \nleq 1$. %
Thus
$(x,y) \notin \passHSTP$. %
\end{example}
\noindent
In constructing the next example, we use Theorem \ref{theorem:pop:mtz}, which states that the partial-ordering model implies the following constraints:
From any directed walk with $H$ arcs that does not start at $r$, at most $H-1$ arcs can be selected.
The next example shows that these constraints do not hold for the assignment model.

\begin{theorem}
    \label{theorem:pop:mtz}
    Let 
    $(x,l,g) \in \ppopHT$(\ref{ppop}),
    let $v_1,\ldots,v_{H+1}$ be a walk with $H$ arcs and $v_1 \ne r$. 
    Then:
    \begin{align}
      \label{pop:xspace:mtz}
      \msum_{i=1}^{H} x_{v_i,v_{i+1}} \le H-1.
    \end{align}
\end{theorem}
\begin{proof}
    Adding up constraints 
    (\ref{constr:edge:forward})
    for the arcs $(v_i,v_{i+1})$ %
    with $1 \le i \le H$ 
    gives
    $\sum_{i=1}^{H} x_{v_i,v_{i+1}} 
      \le \sum_{i=1}^{H} (l_{v_i,i} + g_{i,v_{i+1}}) 
      = l_{v_1,1} +\sum_{i=1}^{H-1}(g_{i,v_{i+1}}+l_{v_{i+1},i+1})+g_{H,v_{H+1}}$.
    Since
    $l_{v_1,1}=0$ for $v_1\ne r$ by (\ref{constr:root}),
    $g_{i,v_{i+1}}+ l_{v_{i+1},i+1}=1$ by (\ref{constr:unambiguous2}), %
    $g_{H,v_{H+1}}=0$ by (\ref{constr:interval}),
    we have 
    $\sum_{i=1}^{H} x_{v_i,v_{i+1}} \le 0 + \sum_{i=1}^{H-1} 1 + 0 = H-1$. %
\end{proof}

\begin{example} [$\passtopopHSTP \setminus \ppopHSTP \ne \emptyset$]	\label{example:ass-pop}
Suppose we have the same \hstp instance as in Example \ref{example:pop-ass}.
Consider the feasible solution $(x,y) \in \passHSTP$ (Fig. \ref{fig:ass}), where: %
\begin{align}
(x_{a,b}, x_{b,c}, x_{c,a},x_{b,a}, x_{a,c}, x_{c,b}) 
&= ( \nicefrac23, \nicefrac23, \nicefrac23,  \nicefrac13, \nicefrac13, \nicefrac13 ),
\label{example:equality:ass-pop}
\\
(y_{r,0},y_{r,1},y_{r,2}) &=(1,0,0),
\notag \\
(y_{v,0},y_{v,1},y_{v,2}) &= (0, \nicefrac 23, \nicefrac 13)
	&& v \in V \setminus \{r\},
\notag 
\end{align}
and the remaining $x$ values are zero.
Via (\ref{connection:asspop}), we 
transform $(x,y)$ %
to $(x,l,g)\in\passtopopHSTP$.
The arcs $(a,b)$, $(b,c)$ build a walk %
with $H=2$ arcs and $a \ne r$. 
As $x_{a,b} + x_{b,c} =\frac43 \nleq 1=H-1$, $x$ violates (\ref{pop:xspace:mtz}),
so $ (x,l,g) \notin \ppopHT$ by Theorem \ref{theorem:pop:mtz}.
The claim follows from $ \ppopHSTP \subseteq \ppopHT $.
\end{example}

\subsection{Comparison by the projection onto the $x$ variable space}
We observe that (\ref{Ahstp}) and (\ref{Phstp}) have the same objective function, which depends only on $x$ variables.
We call a constraint an \emph{$x$-space constraint} if it contains only the $x$ variables.
Let $\Q_X$ denote the projection of polytope $\Q$ onto the $x$ variable space.
To show that (\ref{Phstp}) is stronger than (\ref{Ahstp}), it suffices to prove $\xpopHSTP \subseteq \xassHSTP$ (Figure \ref{fig:projection}).
To this end, 
we first show that
the hop-constrained tree polytopes $\passHT$ (\ref{pass1}) 
and $\ppopHT$ (\ref{ppop}) 
satisfy $\xpopHT\subseteq\xassHT$.
To get the projection $\xassHT$ of $\passHT$,
we first eliminate some $y$ variables using %
(\ref{ass:constr:root})--(\ref{ass:constr:interval})
and then eliminate the remaining $y$ variables %
via Fourier-Motzkin elimination (FME) \cite{schrijver1998theory}. 
\begin{figure}[t]
  \centering
  \scalebox{0.9}{
    \begin{tikzpicture}[
      P/.style= {rectangle, inner sep=1pt, draw},
      E/.style= {->, line width=1},
      Pointer/.style= {->, draw, opacity=0.5},
    ]
    \def\xa{0.2}
    \def\XA{2.5}
    \def\red{red!50!gray}
    \def\blue{blue!70!green}
    \def\red{black}
    \def\blue{black}

    \path (-4,   \XA) coordinate (a1) {};
    \path (-3.7, \xa) coordinate (a2) {};

    \draw[\red, thick, fill=\red ] (0,\xa) rectangle (0.03,\XA); 
    \draw[\red, dotted] (0,\XA) -- (a1);
    \draw[\red, dotted] (0,\xa) -- (a2);
    \draw[\red, fill=\red, fill opacity=0.1 ] 
	(-3,1.7) -- (a1) -- (-5.5,2) -- (-5.5,.4) -- (a2) -- cycle;

    \path (a2)	+(-1.0,1.6) node [text =\red]	{\Large $\passHSTP$}
		+(4.4,1.6) node [text =\red]	{\Large $\passHSTP_X$};
    \path (p2)  +(-2.3,-0.9) node [text =\blue] {\Large $\ppopHSTP$}
		+(-0.7,-1.1) node [text =\blue] {\Large $\ppopHSTP_X$};

    \def\xp{0.3}
    \def\XP{2.3}
    \path (-3,   \XP) coordinate (p1) {};
    \path (-2.2, \xp) coordinate (p2) {};

    \draw[\blue, thick, fill=\blue ] (-0.04,\xp) rectangle (-0.07,\XP); 

    \draw[\blue, dotted] (0,\XP) -- (p1);
    \draw[\blue, dotted] (0,\xp) -- (p2);
    \draw[\blue, fill=\blue, fill opacity=0.1 ] 
	(-1.5,1.5) --  (p1) -- (-4.2,1.5) -- (-3.4,.4) -- (p2) -- cycle;

    \path (0.0,-0.1) edge [Pointer] +(0,3) +(-0.3, 3) node {$x$}; 
    \end{tikzpicture}
  }
  \ifshort
  \caption{Projection of polytopes $\passHSTP$ and $\ppopHSTP$ onto the $x$ variable space}
  \else
  \caption{Polytopes $\ppopHSTP$ and $\passHSTP$ and their projections onto the $x$ variable space }
  \fi
  \label{fig:projection}
\end{figure}
We illustrate the FME by eliminating the variable $x_1$ in the following simple problem. 
The problem consists of two types of constraints $P$ and $N$; %
$P$ contains the variable $x_1$ with the positive sign, 
while $N$ with the negative sign: 
\begin{align*}
  P: &&    a_1x_1 + a_2x_2 + \dots + a_nx_n + a_{n+1}  & \le 0,
  \\
  N: &&   -b_1x_1 + b_2x_2 + \dots + b_nx_n + b_{n+1}  & \le 0.
\end{align*}
Note that $a_1$ and $b_1$ are positive numbers.
Multiplying $P$ by $b_1$ and $N$ by $a_1$ yields 
$
(b_2x_2 + \dots + b_nx_n + b_{n+1} ) \cdot a_1 
  \le a_1 b_1 x_1 \le 
- (a_2x_2 + \dots + a_nx_n + a_{n+1} ) \cdot b_1.
$
Thus $x_1$ can be eliminated, and the constraints $P$ and $N$ can be replaced by the constraint $P \cdot b_1 + N \cdot a_1$, \ie,
$  
(a_2x_2 + \dots + a_nx_n + a_{n+1} ) \cdot b_1
  + 
  (b_2x_2 + \dots + b_nx_n + b_{n+1} ) \cdot a_1 
  \le 0.
$
\ifshort
\else
If we have $k$ constraints of type $P$ and $k'$ constraints of type $N$, then we combine each constraint of type $P$ with each constraint of type $N$, so we get $k \cdot k'$ new constraints.
\fi

\begin{remark} \label{remark:fme:intro}
If all coefficients in constraints of a problem are integral, then applying the FME on the problem produces  new constraints, which are a weighted sum of constraints of input problem such that the weights are nonnegative integers.
\end{remark}

\noindent
We first fix all $y$ variables corresponding to $r$ via equations (\ref{ass:constr:root}), (\ref{ass:constr:root:i}). 
For each remaining node $v\in V \setminus \{r\}$, we eliminate $y_{v,0}$ using (\ref{ass:constr:Vroot}) and $y_{v,1}$ using (\ref{ass:constr:interval}):
\begin{align}
   y_{v,1}=1-(y_{v,2} +\dots+ y_{v,H})  && \forall v \in V \setminus \{r\}.
  \label{replacement:y_v1}
\end{align}
Concerning (\ref{ass:constr:arcdir}) and (\ref{ass:constr:arcdirH}), there are three cases w.r.t. an edge $uv \in E$:
 In case $r=u$, setting
 (\ref{ass:constr:root}) 
 and (\ref{replacement:y_v1})
 to (\ref{ass:constr:arcdir}) 
 for $i=0$ yields
 the following set of constraints, which is denoted by $S_0$:
   \begin{align*}
    S_0:
    && (y_{v,2} +\dots+ y_{v,H}) + x_{r,v}	    & \le 1   && 
    \forall v \in V \setminus \{r\}.
   \end{align*}
 Setting 
 (\ref{ass:constr:root:i})
 to (\ref{ass:constr:arcdir}) 
 for $i = 1,\dots,H-1$ 
 gives
 $- y_{v,i+1} + x_{r,v} \le 1$, 
 which is redundant as it is implied by
 $-y_{v,i+1} \le0$ and $x_{r,v} \le1$.
 Constraints
 (\ref{ass:constr:arcdirH}) have the form $x_{r,v}\le1$, too.

  \noindent    
  In case $r \notin \{u,v\}$,
    (\ref{ass:constr:arcdir}) are redundant for $i=0$,
    since setting 
    (\ref{ass:constr:Vroot}) to 
    (\ref{ass:constr:arcdir}) 
    gives
    $ - y_{v,1} + x_{u,v} \le 1 $, 
    which is %
    implied by
    $-y_{v,1} \le 0$ and $x_{u,v}\le1$.
    Setting 
    (\ref{replacement:y_v1})
    to (\ref{ass:constr:arcdir}) 
    for $i=1$ 
    gives the set 
    $S_1$:
    \begin{align*}
	S_1: &&
       -(y_{u,2} +\dots+y_{u,H}) - y_{v,2} + x_{u,v} & \le 0 &&
	\forall uv \in E, r \notin \{u,v\}.
    \end{align*}
    Inequalities (\ref{ass:constr:arcdir}) 
    for $i \in \{2,\cdots,H-1\}$ 
    remain unchanged and are denoted by 
    $S_i$.
    The constraints (\ref{ass:constr:arcdirH}) also remain unchanged 
    and are denoted by 
    $S_H$.
    \begin{align*}
      S_i: &&
      y_{u,i} - y_{v,i+1} + x_{u,v} \le 1 
      &&& \forall uv \in E, r \notin \{u,v\},\ i = 2,\cdots,H-1.
      \\
      S_H: &&
      y_{u,H} + x_{u,v} \le 1 
      &&& \forall uv \in E, r \notin \{u,v\}.
    \end{align*}
 In case $r=v$,
   for $i=0$,
   setting 
   (\ref{ass:constr:root:i})
   and 
   (\ref{ass:constr:Vroot})
   to (\ref{ass:constr:arcdir}) 
   implies $x_{u,r} \le 1$. 
   Setting 
   (\ref{ass:constr:root:i})
   to (\ref{ass:constr:arcdir}) 
   and 
   (\ref{ass:constr:arcdirH}) 
   for $i \in \{1,\dots,H\}$
   gives
   $y_{u,i} + x_{u,r} \le 1$.
   We replace %
   $y_{u,i} + x_{u,r} \le 1$
   by stronger constraints
     \begin{align}
      x_{u,r} & \le 0 &&
      \forall u \in V \setminus \{r\},
      \label{constr:arc:ur}
     \end{align}
 and denote the resulting polytope by $\passHTstr$.
Concerning (\ref{ass:vars}),
notice that $y_{v,i}\le1$ is redundant as it is implied by 
$y_{v,i}\ge0$ and (\ref{ass:constr:interval}).
Setting (\ref{replacement:y_v1}) 
to $y_{v,1} \ge 0$ for each $v \in V \setminus \{r\}$
yields 
$\sum_{j=2}^Hy_{v,j} \le 1$, 
which is implied by constraints $S_0$.
For $i \ge 2$, we rewrite constraints $y_{v,i}\ge0$:
\begin{align}
  - y_{v,i}		& \le 0    &&	\forall v \in V \setminus \{r\},\ 
    i \in \{2,\cdots,H\}.
  \label{Q2}
\end{align}
The polytope $\passHTstr$ strengthened by (\ref{constr:arc:ur}) has the form
\begin{gather} \label{pass2}
  \passHTstr = \{ (x,y) \colon (x,y) \mbox{ satisfy } 
  S_0,\dots,S_H, 
  (\ref{constr:indeg})-(\ref{constr:xvars}),
  (\ref{constr:arc:ur}),
  (\ref{Q2})
  \}. 
\end{gather}
\ifshort
Set $S_0$ contains exactly one constraint with variable $x_{r,v}$ for each neighbor      $v$ of $r$, while $S_i$ with $i \ge 1$ has two constraints for each $uv \in E$ with $r \notin \{u,v\}$, one for arc $(u,v)$, one for $(v,u)$. 
By $S_i^{(u,v)}$, we denote the $S_i$ constraint for arc $(u,v)$.
\else
\begin{remark}   %
\label{remark:Si:1}
Set $S_0$ contains exactly one constraint with variable $x_{r,v}$ for each neighbor      $v$ of $r$, while $S_i$ with $i \ge 1$ has      two constraints for each $uv \in E$ with $r \notin \{u,v\}$, one for arc $(u,v)$, one for $(v,u)$. 
By $S_i^{(u,v)}$, we denote the $S_i$ constraint for arc $(u,v)$.
Note that if $i=0$, then $u=r$ and $v \ne r$, otherwise $ u \ne r \ne v$.
\end{remark}
\fi
\ifshort
\else
\begin{figure}[bt]
  \centering
  \scalebox{0.8}{
    \begin{tikzpicture}[
      x=1.0 cm, y=0.9 cm,
      V/.style= {circle, inner sep=4pt, draw, line width=1pt},
      P/.style= {rectangle, inner sep=1pt, draw},
      E/.style= {->, line width=1},
      L/.style  = {bend left=7pt},
      R/.style  = {bend right=7pt},
      Pointer/.style= {->, draw, opacity=0.5},
    ]

      \path (0,1.0)	node [V] (d) {$b$};
      \path (2,1.0)	node [V] (e) {$e$};
      \path (1.5,2.0)	node [V] (a) {$a$};
      \path (4,1.0)	node [V] (c) {$c$};
      \path (6,1.0)	node [V] (t) {$t$};
      \path (3,0.0)	node [V] (h) {$r$};

      \path (a) edge [E] node [above] {$W_1$}  (c);
      \path (c) edge [E,L] (t);

      \path (d) edge [E, dashed] node [below] {$W_2$}  (e);
      \path (e) edge [E, dashed] (c);
      \path (c) edge [E,R, dashed] (t);

      \path (h) edge [E, dotted] node [below] {$W_3$}  (t);
    \end{tikzpicture}
  }
  \caption{Example walks 
    $W_1(a,c,t)$, $W_2(b,e,c,t)$, $W_3(r,t)$ have the same end $t$, such that for each
    $i \ne j \in \{1,2,3\}$, there is a node $v$ 
    with $|\inarc_{W_i}(v) \cup \inarc_{W_j}(v)| \ge 2$, 
    \eg, $\inarc_{W_1}(c) \cup \inarc_{W_2}(c) = \{(a,c), (e,c)\}$.
  }
  \label{fig:singlesink}
\end{figure}
\fi
\begin{lemma} %
    \label{lemma:asw}
    Let $\W$ be a set of %
    walks that end in the same node $t$,
    such that 
    for each two $W \ne W' \in \W$, there is a node $v$
    with 
    \ifshort
    $|\inarc_W(v) \cup \inarc_{W'}(v)| \ge 2$.
    \else
    $|\inarc_W(v) \cup \inarc_{W'}(v)| \ge 2$ (Fig. \ref{fig:singlesink}).
    \fi
    Then %
    (\ref{constr:indeg}) and (\ref{constr:xvars}) imply:
    \begin{align} \label{asw0}
        \msum_{W \in \W} x(W) 
        & \le 
        \big( \msum_{W \in \W} \len W \big) - |\W|+1.
    \end{align}
\end{lemma}
\ifshort
For space reasons, we postpone the proof of \autoref{lemma:asw} to the full version \cite{jabrayilov2020hopconstrained}.
\else
\begin{proof} %
    \label{proof:lemma:asw}
    We show it by induction on $|\W|$.
    For $|\W|=1$, 
    constraint (\ref{asw0}) has the form
    $x(W) \le \len W -1+1=\len W$, 
    which is implied by (\ref{constr:xvars}).
    For $|\W| \ge 2$, there are two cases:
    In the first case, there is a walk 
    $W':=a_1,\dots,a_m \in \W$ that contains a node $v$ 
    with $|\inarc_{W'}(v)| \ge 2$.
    By the induction hypothesis, (\ref{constr:indeg}) and (\ref{constr:xvars}) imply:
    \begin{align}
	\sum_{W \in \W \setminus W'} x(W) 
	& \le 
	\Big( \sum_{W \in \W \setminus W'} \len W  \Big)
	-|\W \setminus W'|+1
	=
	\Big( \sum_{W \in \W \setminus W'} \len W  \Big)
	-|\W|+2.
	    \label{ineqTemp1}
    \end{align}
    Let $a_j \ne  a_k \in \inarc_{W'}(v)$. 
    We have
    $\sum_{i\in \{j,k\}} x_{a_i} \le 1$ 
    by (\ref{constr:indeg})
    and
    $\sum_{i\in\{1,\dots,m\}\setminus \{j,k\}} x_{a_i} \le m-2$ 
    by 
    (\ref{constr:xvars}).
    Thus $x(W') =
    \sum_{i\in \{j,k\}} x_{a_i} 
    + \sum_{i\in\{1,\dots,m\}\setminus \{j,k\}} x_{a_i} \le m-1 =\len {W'}-1$
    is implied 
    by (\ref{constr:indeg}) and 
    (\ref{constr:xvars}).
    Adding 
    $x(W') \le \len {W'}-1$ and 
    (\ref{ineqTemp1}) 
    proves the claim.

    In the second case, %
    for any walk
    $W:=v_0,a_1,v_1,\allowbreak \dots,a_{\len W},v_{\len W} \in \W$ and
    any $i \in \{1,\dots,\len W\}$, we have 
    $|\inarc_{W}(v_i)| \le 1$,
    \ie,
    $\inarc_W(v_i)=\{a_i\}$. Then:
    \begin{align}
        x(W) = 
        \sum\nolimits_{i=1}^{\len W} x_{a_i} 
        = \sum\nolimits_{i=1}^{\len W} x(\inarc_W(v_i)).
        \label{lemmaStw1}
    \end{align}
    Consider the set %
    $\inarc_{\W}(t):=\bigcup_{W\in\W}\inarc_W(t)$. %
    If $|\inarc_{\W}(t)|=1$,
    \ie, all walks in $\W$ 
    ends in the same arc,
    then we traverse the walks backward till we find a node $t'$ with $\inarc_W(t')\ge 2$:
    We initialize $t':=t$. 
    While 
    $|\inarc_{\W}(t')|=1$, say 
    $\inarc_{\W}(t)=\{w,t'\}$, we set $t':=w$.
    This way, we get a node $t'$ with 
    $|\inarc_{\W}(t')|\ge2$,
    since for each two $W \ne W' \in \W$, there is a node $v$ 
    with 
    $|\inarc_W(v) \cup \inarc_{W'}(v)| \ge 2$.
    Let $\inarc_{\W} (t') = \{(w_1,t'),\hdots,(w_k,t')\}$ with $k\ge2$.
    Constraints (\ref{constr:indeg}) imply:
    \begin{align}
    x(\inarc_{\W}(t')) = \sum\nolimits_{i=1}^k x_{w_i,t'} \le 1.
    \label{lemmaIneqDeltaInT}
    \end{align}
    Let $\W_i \subset \W$ for $i \in \{1,\hdots,k\}$ be the set of walks
    that contain the walk $w_i,t',\dots,t$ as a suffix. 
    Then $\inarc_{\W_i} (t') = \{(w_i,t')\}$.
    Thus
    $\sum_{W \in \W_i}  x(\inarc_W(t'))= |\W_i| \cdot x_{v_i,t'}$.
    Let $l$ be the length of walk $t',\dots,t$,
    and let
    $W:=(v_0,v_1,\dots,v_{\len W}) \in \W_i$. 
    Then $t'=v_{\len W-l}$. %
    From (\ref{lemmaStw1}) follows: %
    \begingroup
    \allowdisplaybreaks
    \begin{align}
        \sum_{W \in \W_i} x(W) 
        & = \sum_{W:=(v_0,\dots,v_{\len W}) \in \W_i} 
            \sum_{j = 1}^{\len W} x \big( \inarc_W(v_j) \big)
            \notag
        \\
        & = \sum_{W:=(v_0,\dots,v_{\len W}) \in \W_i} 
            \Big( 
            x \big( \inarc_W(t') \big)
         + \sum_{j \in \{1,\dots,\len W\} \setminus \{\len W-l\}} x \big( \inarc_W(v_j) \big)
            \Big)
            \notag
        \\
        & = |\W_i| \cdot x_{w_i,t'} 
            + 
            \sum_{W:=(v_0,\dots,v_{\len W}) \in \W_i} 
            \sum_{j \in \{1,\dots,\len W\} \setminus \{\len W-l\}} x \big( \inarc_W(v_j) \big).
            \label{lemmaStw2}
    \end{align}
    \endgroup
    Since $|\W_i| < |\W|$, by the induction hypothesis, (\ref{constr:indeg}) and (\ref{constr:xvars}) imply: 
    \begin{align}
      \sum\nolimits_{W \in \W_i} x(W) 
      & \le 
      \Big( \sum\nolimits_{W \in \W_i} \len W \Big)
      -|\W_i|+1.
      \label{lemmaStw3}
    \end{align}
    Constraints 
    (\ref{constr:indeg}) 
    affect $\sum_{W \in \W_i} x(W)$ 
    if and only if there is a node $w$ with
    $|\inarc_{\W_i} (w)| \ge 2$,
    \ie, 
    there are two distinct arcs 
    $a, a' \in \inarc_{\W_i} (w)$,
    as then 
    $\sum_{W \in \W_i} x(W)$
    includes the summands
    $x_{a},\ x_{a'}$,
    such that
    $x_{a}+x_{a'} \le 1$ by (\ref{constr:indeg}).
    From 
    $|\inarc_{\W_i}(t')|=|\{(w_i,t')\}|=1$
    follows $w \ne t'$.
    The variables corresponding to arcs $\inarc_{\W_i} (w)$ with $w \ne t'$
    are in the second summand of (\ref{lemmaStw2}).
    That is, constraints 
    (\ref{constr:indeg}) affect only 
    the second summand of (\ref{lemmaStw2}),
    while 
    (\ref{constr:xvars})
    affect both summands of (\ref{lemmaStw2}).
    Particularly, 
    (\ref{constr:xvars})
    imply for the first summand that
    $|\W_i| \cdot x_{w_i,t'}  \le |\W_i|$.
    Then
    (\ref{constr:indeg}) 
    and
    (\ref{constr:xvars})
    imply (\ref{lemmaStw3}),
    if and only if  
    (\ref{constr:indeg}) 
    and 
    (\ref{constr:xvars})
    imply 
    for the second summand of (\ref{lemmaStw2}) that:
    \begin{align}
      \sum_{W:=(v_0,\dots,v_{\len W}) \in \W_i} 
      \sum_{j \in \{1,\dots,\len W\} \setminus \{\len W-l\}} x \big( \inarc_W(v_j) \big)
      \le 
       \Big( \sum_{W \in \W_i} \len W \Big) 
	- |\W_i| +1
	- |\W_i|. 
      \label{lemmaStw4}
    \end{align}
    It follows:
    \begingroup
    \allowdisplaybreaks
    \begin{align}		    
	\sum\nolimits_{W \in \W} x(W) 
	& =
	\sum\nolimits_{i=1}^k
	    \sum\nolimits_{W \in \W_i} x( W ) 
	\notag
	\\
	& \le
	\sum\nolimits_{i=1}^k
	\Big( 
	    |\W_i| \cdot x_{w_i,t'} 
	    +  \big( \sum\nolimits_{W \in \W_i} \len W \big)
	    - |\W_i| +1 -|\W_i|
	\Big) 
	\label{lemmaStw6}
	\\
	& \le
	\sum\nolimits_{i=1}^k
	\Big( 
	    \big( \sum\nolimits_{W \in \W_i} \len W \big)
	    - |\W_i| 
	    + x_{w_i,t'} 
	\Big) 
	\label{lemmaStw7}
	\\
	& \le 
	\Big( \sum\nolimits_{W \in \W} \len W \Big)
	-|\W|
	+
	1,
	\label{lemmaStw8}
    \end{align}
    \endgroup
    where %
    (\ref{lemmaStw6}) follows from 
    (\ref{lemmaStw2})
    and (\ref{lemmaStw4}). %
    The reason for (\ref{lemmaStw7}) is that
    $ |\W_i| \cdot x_{w_i,t'} - (|\W_i|-1 ) = x_{w_i,t'} + ( x_{w_i,t'}-1) ( |\W_i|-1 ) \le x_{w_i,t'} $
    as $x_{w_i,t'}\le1$ and $|\W_i|\ge1$. 
    Inequality (\ref{lemmaStw8}) %
    follows from
    $ \sum_{i=1}^k \sum_{W \in \W_i} \len W = \sum_{W \in \W} \len W$
    and %
    (\ref{lemmaIneqDeltaInT}).
\end{proof}
\fi
\begin{theorem} \label{theorem:hstp:fme}
  Hop-constrained 
  tree polytopes $\passHT$ (\ref{pass1}) 
  and $\ppopHT$ (\ref{ppop}) 
  \ifshort
  meet
  \else
  satisfy 
  \fi
  $\xpopHT\subseteq\xassHT$.
\end{theorem}
\begin{proof} %
Polytopes $\passHT$ (\ref{pass1}) and $\passHTstr$ (\ref{pass2}) satisfy $\passHTstr \subseteq \passHT$. 
Then $\xassHTstr \subseteq \xassHT$.
So it is enough to show $\xpopHT \subseteq \xassHTstr$.
$\passHTstr$ constraints (\ref{constr:indeg})--(\ref{constr:xvars}) 
also hold for $\xpopHT$, by definition of $\ppopHT$.
Moreover, adding $\ppopHT$ constraints
$-g_{0,r} = 0$                (\ref{constr:root}),
$-l_{u,1}=0$                  (\ref{constr:interval}),
$-l_{u,0}+l_{u,1} \ge 0$      (\ref{constr:unambiguous1}), and 
$l_{u,0}+g_{0,r} \ge x_{u,r}$ (\ref{constr:edge:forward})
gives $x_{u,r} \le 0$         (\ref{constr:arc:ur}),
so (\ref{constr:arc:ur}) hold for $\xpopHT$, too.
So we consider only the $x$-space constraints $X$, generated by applying the FME on the remaining $\passHTstr$ constraints $S_0,\dots,S_H$ and (\ref{Q2}).

In case $H=1$, the set (\ref{Q2}) is empty as it is defined for $H\ge 2$, while $S_0$ and $S_1$ 
has the form:
\begin{align*}
S_0: &&	    x_{r,v}	& \le 1   && 
\forall v \in V \setminus \{r\},
\\
S_1: &&	    x_{u,v}	& \le 0	  && \forall uv \in E, r \notin \{u,v\}.
\end{align*}
$\xpopHT$ constraints (\ref{constr:xvars}) imply $S_0$. 
By Theorem \ref{theorem:pop:mtz}, $\xpopHT$ constraints (\ref{pop:xspace:mtz}) imply $S_1$,
since
the arc $(u,v)$ with $r \notin \{u,v\}$ is a directed walk $u,v$ with $H=1$ arc such that $u\ne r$.
Thus $\xpopHT \subseteq \xassHTstr$.

Otherwise, $H \ge 2$. %
By Remark
\ref{remark:fme:intro}, %
$X$ is a weighted sum 
of some constraints selected from sets $S_0,\dots,S_H$ and (\ref{Q2}),
where the weights are nonnegative integers.
We can interpret an integral weight $w$ of a constraint $C$ as meaning that
constraint $C$ is selected $w$ time.
Thus, $X$ can be shown as unweighted sum 
$C_1 + \dots + C_l$
of some not necessarily distinct constraints $C_1,\dots,C_l$
selected from sets $S_0,\dots,S_H$ and (\ref{Q2}).
Moreover, since in the constraints $C_1,\dots,C_l$, all coefficients are integral, this is true for $X$, too.
Let $s_i$ be the number of the selected $S_i$ constraints, and let $s:=(s_0+\dots+s_H)$. 
Since $X$ does not contain any $y$ variable, it has the form: 
\begin{align}
  X: && 
  x_{a_1} + \dots + x_{a_{s}} & \le s-s_1, 
  \label{X}
\end{align}
where $a_1,\dots,a_s$ are some not necessarily distinct arcs.
Let $J:=\{1,\dots,s\}$ and let $J_0,\dots,J_H$ be the partition of $J$ such that $j \in J_i$ 
\ifshort
if $x_{a_j}$ is originated from $S_i^{a_j}$.
\else
if $x_{a_j}$ is originated from $S_i^{a_j}$ with $0 \le i \le H$;
in this case, we write $\S(j)=S_i^{a_j}$.
\fi
We show that there is a partition $\I$ of $J$, such that  (\ref{constr:indeg}), (\ref{constr:xvars}), and (\ref{pop:xspace:mtz}) imply:
\begin{align}
  \msum_{i \in I} x_{a_i} \le |I \setminus J_1 |
  &
  &
  \mbox{for all } I \in \I.
  \label{index:partition}
\end{align}
It follows then $\xpopHT$ constraints (\ref{constr:indeg}), (\ref{constr:xvars}), and (\ref{pop:xspace:mtz}) imply $X$ (\ref{X}), as
\begin{align*}
    x_{a_1} + \dots + x_{a_s} 
    = 
    \msum_{I\in\I} \sum\nolimits_{i \in I} x_{a_i}
    \le \msum_{I\in \I}|I\setminus J_1|
    = \big( \msum_{I\in \I}|I| \big) - |J_1|
    = s - s_1,
\end{align*}
and thus $\xpopHT \subseteq \xassHTstr$.
\ifshort
We construct $\I$ such that it contains three types of sets:
\else
We construct $\I$ (\ref{index:partition}) such that it contains three types of sets:
\fi

\ifshort
  \noindent\textbf{Type 1}:
  Each set $I:=\{j_1,\dots,j_H\}$
  of the first type
  has exactly $H$ elements 
  such that 
  $|I \cap J_1|=1$ and
  $a_{j_1},\dots,a_{j_H}$ build a walk with $H$ arcs
  that does not start at $r$.
  From (\ref{pop:xspace:mtz}) follows 
  $\sum_{i \in I} x_{a_i} %
  \le H-1 = |I\setminus J_1|$.
  Thus $I$ satisfies (\ref{index:partition}).

  \noindent\textbf{Type 2}:
  Let $I:=\{j_1,\dots,j_m\}$ be a set of the second type.
  The arcs $a_{j_1},\dots,a_{j_m}$ build $|I\cap J_1|+1$ 
  walks $\W := \{W_1,\dots,W_{|I\cap J_1|+1}\}$ that end in the same node
  such that 
  each two $W \ne W' \in \W$ has a common node $v$
  \ifshort
  with $|\inarc_W(v) \cup \inarc_{W'}(v)| \ge 2$.
  \else
  with $|\inarc_W(v) \cup \inarc_{W'}(v)| \ge 2$ (Figure \ref{fig:singlesink}).
  \fi
  By Lemma \ref{lemma:asw}, constraints 
  (\ref{constr:indeg}) and (\ref{constr:xvars}) imply 
  $\sum_{i \in I} x_{a_i} = 
      \sum_{W \in \W} x(W) \le ( \sum_{W \in \W} \len W ) 
      - |\W| +1
      = |I| - |I\cap J_1| 
      = |I\setminus J_1|$,
  so $I$ satisfies (\ref{index:partition}).
  \\
  \noindent\textbf{Type 3}:
  There is only one set $I$ of the last type. $I$ does not contain any element from $J_1$, \ie, $I \setminus J_1=I$.
  By (\ref{constr:xvars}), 
  $\sum_{i \in I} x_{a_i} \le |I| = |I\setminus J_1|$,
  so $I$ meets (\ref{index:partition}).
\else
\begin{description}
  \item[Type 1:]
  Each set $I:=\{j_1,\dots,j_H\}$
  of the first type
  has exactly $H$ elements 
  such that 
  $|I \cap J_1|=1$ and
  $a_{j_1},\dots,a_{j_H}$ build a walk with $H$ arcs
  that does not start at $r$ (Fig. \ref{figSetType1}).
  From (\ref{pop:xspace:mtz}) follows 
  $\sum_{i \in I} x_{a_i} %
  \le H-1 = |I\setminus J_1|$.
  Thus $I$ satisfies (\ref{index:partition}).

  \item[Type 2:]
  Let $I:=\{j_1,\dots,j_m\}$ be a set of the second type.
  The arcs $a_{j_1},\dots,a_{j_m}$ build $|I\cap J_1|+1$ 
  walks $\W := \{W_1,\dots,W_{|I\cap J_1|+1}\}$ that end in the same node
  such that 
  each two $W \ne W' \in \W$ have a common node $v$
  with $|\inarc_W(v) \cup \inarc_{W'}(v)| \ge 2$.
  By Lemma \ref{lemma:asw}, constraints 
  (\ref{constr:indeg}) and (\ref{constr:xvars}) imply 
  $\sum_{i \in I} x_{a_i} = 
      \sum_{W \in \W} x(W) \le ( \sum_{W \in \W} \len W ) 
      - |\W| +1
      = |I| - |I\cap J_1| 
      = |I\setminus J_1|$.
  Hence $I$ fulfills (\ref{index:partition}).
  \ifshort
  \else
  Figure \ref{figSetType2} illustrates an example for the second type of sets, namely 
  the set $I:=\{j_1,\dots,j_6\}$ with $I \cap J_1 = \{j_1,j_3\}$. %
  The arcs $a_{j_1},\dots,a_{j_6}$ build $|I \cap J_1|+1=3$ walks $\W=\{W_1,W_2,W_3\}$ (solid, dashed and dotted) that end in the same node, and each two walks $W \ne W' \in \W$ have a common node $v$ with at least two distinct incoming arcs, \ie, $|\inarc_W(v) \cup \inarc_{W'}(v)| \ge 2$. 
  By Lemma \ref{lemma:asw}, constraints 
  (\ref{constr:indeg}) and (\ref{constr:xvars}) imply 
  $x_{a_{j_1}}+\dots+x_{a_{j_6}} \le |I| - |\W|+1 = 6-3+1=|I| - (|I\cap J_1|+1) +1 = | I \setminus J_1| $.
  \fi
  \item[Type 3:]
  There is only one set $I$ of the last type, which does not contains any element from $J_1$, \ie, $I \setminus J_1=I$.
  By (\ref{constr:xvars}), 
  $\sum_{i \in I} x_{a_i} \le |I| = |I\setminus J_1|$,
  so $I$ fulfills (\ref{index:partition}).
\end{description}
\fi

\ifshort
\else
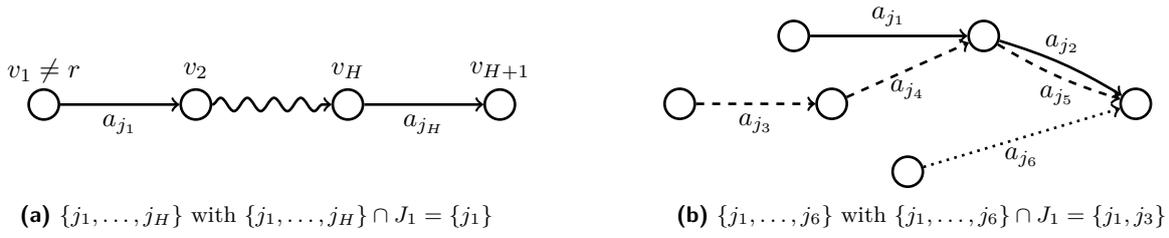
\begin{figure}[bt]
  \centering
  \def\scale{1.0}

  \captionsetup[subfigure]{justification=centering}
  \begin{subfigure}[b]{0.44\textwidth}
    \scalebox{\scale}{
    \begin{tikzpicture}[
      x=1.0 cm, y=0.9 cm,
      V/.style= {circle, inner sep=4pt, draw, line width=1pt},
      P/.style= {rectangle, inner sep=1pt, draw},
      snake/.style= {->, decorate, decoration=snake, line width=1},
      E/.style= {->, line width=1},
      L/.style  = {bend left=7pt},
      R/.style  = {bend right=7pt},
      Pointer/.style= {->, draw, opacity=0.5},
      ]
      \path (0,1.0)	node [V] (d) {} +(0,0.5) node {$v_1\ne r$};
      \path (2,1.0)	node [V] (e) {} +(0,0.5) node {$v_2$};
      \path (4,1.0)	node [V] (c) {} +(0,0.5) node {$v_H$};
      \path (6,1.0)	node [V] (t) {} +(0,0.5) node {$v_{H+1}$};
      \path (d) edge [E] node [below] {$a_{j_1}$}  (e);
      \path (e) edge [snake]  (c);
      \path (c) edge [E] node [below] {$a_{j_H}$}  (t);

      \path (0,-0.1) node {};
    \end{tikzpicture}
    }
  \subcaption{
    $\{j_1,\dots,j_H\}$ with 
    $\{j_1,\dots,j_H\} \cap J_1 = \{j_1\}$
  }
  \label{figSetType1}
  \end{subfigure}
  \hspace{42pt}
  \captionsetup[subfigure]{justification=centering}
  \begin{subfigure}[b]{0.44\textwidth}
    \scalebox{\scale}{
    \begin{tikzpicture}[
      x=1.0 cm, y=0.9 cm,
      V/.style= {circle, inner sep=4pt, draw, line width=1pt},
      P/.style= {rectangle, inner sep=1pt, draw},
      E/.style= {->, line width=1},
      L/.style  = {bend left=7pt},
      R/.style  = {bend right=7pt},
      Pointer/.style= {->, draw, opacity=0.5},
      ]
      \path (0,1.0)	node [V] (d) {};
      \path (2,1.0)	node [V] (e) {};
      \path (1.5,2.0)	node [V] (a) {};
      \path (4,2.0)	node [V] (c) {};
      \path (6,1.0)	node [V] (t) {};
      \path (3,0.0)	node [V] (h) {};
      \path (a) edge [E]	    node [above] {$a_{j_1}$}  (c);
      \path (c) edge [E,L]          node [above] {$a_{j_2}$}  (t);
      \path (d) edge [E, dashed]    node [below] {$a_{j_3}$}  (e);
      \path (e) edge [E, dashed]    node [below] {$a_{j_4}$}  (c);
      \path (c) edge [E, dashed, R] node [below] {$a_{j_5}$}  (t);
      \path (h) edge [E, dotted]    node [below] {$a_{j_6}$}  (t);
    \end{tikzpicture}
    }
  \subcaption{
    $\{j_1,\dots,j_6\}$ with $\{j_1,\dots,j_6\} \cap J_1 = \{j_1,j_3\}$
    }
  \label{figSetType2}
  \end{subfigure}
  \caption{
    (\subref{figSetType1}) A set of the first type $\{j_1,\dots,j_H\}$ 
    and 
    (\subref{figSetType2}) a set of the second type $\{j_1,\dots,j_6\}$ 
  }
\end{figure}
\fi

\ifshort
The details of the construction $\I$ (\ref{index:partition}) are given in the full version \cite{jabrayilov2020hopconstrained}.
\else
\def\propertyDefinition{i}
\def\propertySuffix{ii}
\def\propertyDisjoint{iii}

\begingroup
\allowdisplaybreaks
\begin{algorithm2e}[t]
    \caption{Construction of the partition $\I$}
    \label{algoPartitionI}
    \setstretch{0.9}
    \scriptsize
    \LinesNumbered 

    \nl
    $\I \leftarrow \{\}$
    \label{algoInitPartitionI}

    \nl
    $\Jused \leftarrow \{\}$
    \label{algoInitJused}

    \tcc{Phase 1: Search for the first type of sets in $\I$}

    \ForEach {$j_1 \in J_1$ \label{phase1Loop} } {
      \If {there is $\{j_1,\dots,j_H\} \subseteq J \setminus \Jused$
	such that 
        $(j_1,\dots,j_H) \in J_1 \times \dots \times J_H$,
	and $(a_{j_1},\dots,a_{j_H})$ is a walk
	\label{phase1Select} }
      {

	$\I \leftarrow \I \cup \{\{ j_1,\dots,j_H \}\}$
        \hfill \tcc{$\{ j_1,\dots,j_H \}$ is a set of the first type}
	\label{phase1ExtendPartitionI}

	$\Jused \leftarrow \Jused \cup \{ j_1,\dots,j_H \}$
	\label{phase1Jused}

      }
    }

    \tcc{Phase 2: Search for the second type of sets in $\I$}
    $J^* \leftarrow \Jused$
    \label{phase2Jstar}

    \ForEach {$j_0 \in J_0$ \label{phase2Loop1} } {
      $I_{j_0} \leftarrow \{j_0\}$
      \label{phase2SetIj0}

      $\Y_{j_0} \leftarrow \{\}$
      \label{phase2Y0vars0}
    }

    \For {$k = 1,\dots,|J_1 \setminus J^*|$ \label{phase2Loop2} }
      {
      Find a $j_0 \in J_0$ and a 
      set $\{j_1,\dots,j_l\} \subseteq J \setminus \Jused$ 
      of minimum cardinality with 
      $(j_1,\dots,j_l) \in J_1 \times \dots \times J_l$
      such that 
      $(a_{j_1},\dots,a_{j_l})$ is a walk,
      $\head(a_{j_0})=\head(a_{j_l})$ 
      and
      $y_{\head(a_{j_0}),l+1} \notin \Y_{j_0}$.
      \label{phase2Select1}

      $I_{j_0} \leftarrow I_{j_0} \cup \{ j_1,\dots,j_l \}$
      \label{phase2ExtendSetIj0}

      $\Jused \leftarrow \Jused \cup \{ j_1,\dots,j_l \}$
      \label{phase2Jused1}

      $\Y_{j_0} \leftarrow \Y_{j_0} \cup \{y_{\head(a_{j_0}),l+1}\}$
      \label{phase2Y0vars1}

      \For{$i = 2,\dots,H-l$ \label{phase2Loop3} }{
	\If {$y_{\head(a_{j_0}),l+i} \notin \Y_{j_0}$ and 
	there is 
	$\{{j'_1},\dots,{j'_l}\} \subseteq J \setminus \Jused$
	with
	$(j'_1,\dots,j'_l) \in J_i \times \dots \times J_{l+i-1}$ 
	and
	$(a_{j'_1},\dots,a_{j'_l}) = (a_{j_1},\dots,a_{j_l})$%
	\label{phase2Select2}}%
	{

	  $\Jused \leftarrow \Jused \cup \{j'_1,\dots,j'_l\}$
	  \label{phase2Jused2}

	  $\Y_{j_0} \leftarrow \Y_{j_0} \cup \{y_{\head(a_{j_0}),l+i}\}$
	  \label{phase2Y0vars2}
	}
      }
    }

    \ForEach {$j_0 \in J_0$ \label{phase2Loop4} } {
      \If{ $|I_{j_0} \cap J_1| > 0$  \label{phase2IsSecondTypeSet} }{
	$\I \leftarrow \I \cup \{I_{j_0}\} $
        \hfill \tcc{$I_{j_0}$ is a set of the second type}
      }
      \label{phase2ExtendPartitionI}
    }

    \tcc{Phase 3: The remaining set $J \setminus \cup_{I \in \I} I$ is the only set of the third type.}

    $\I \leftarrow \I \cup \big\{ J \setminus \bigcup_{I \in \I} I \big\}$
    \label{phase3ExtendPartitionIBySetType3}
\end{algorithm2e}
\endgroup

\noindent
To construct $\I$ we use Algorithm \ref{algoPartitionI}.
The algorithm initialize $\I$ on line \ref{algoInitPartitionI} as an empty set-family.
To ensure that the sets in $\I$ are disjoint, we manage the set $\Jused$ of used elements
and choose the elements from the remaining set $J \setminus \Jused$. 
The set $\Jused$ is initially empty (line \ref{algoInitJused}). %

The algorithm has three phases, one phase for each set type.
The first phase (lines \ref{phase1Loop}-\ref{phase1Jused}) 
searches for the first type of sets.
The loop spanning lines \ref{phase1Loop}-\ref{phase1Jused} checks on line \ref{phase1Select} for each $j \in J_1$ 
if there is a set 
$I:=\{j_1,\dots,j_H\} \subseteq J \setminus \Jused$ 
such that $(j_1,\dots,j_H) \in J_1 \times \dots \times J_H$ and $W:=a_{j_1},\dots,a_{j_H}$ is a walk.
If yes, 
we add $I$ to $\I$ (line \ref{phase1ExtendPartitionI})
as $I$ is a set of the first type.
The reason is the following:
We have $I \cap J_1= \{j_1\}$, so $|I \cap J_1|= 1$.
Moreover, as $j_1 \in J_1$, $x_{a_{j_1}}$ is originated from $S_1$.
Then, by Remark \ref{remark:Si:1}, $a_{j_1}$ is not incident to $r$. Thus $W$ is a walk with $H$ arcs that does not start at $r$ (Fig. \ref{figSetType1}).
Line \ref{phase1Jused} updates the set $\Jused$. %

The second phase (lines \ref{phase2Jstar}-\ref{phase2ExtendPartitionI}) searches for the second type of sets. 
Consider the set $J^*= \Jused$ (line \ref{phase2Jstar}) of used elements in the first phase.
As no $j_1 \in J_1 \setminus J^*$ meets the conditions in line \ref{phase1Select}, there is no subset of $J \setminus J^*$ that is a first type set.  
We divide the remaining elements of $J_1 \setminus J^*$ between the second type of sets.
We construct these sets such that each of them has a unique element $j_0$ from set $J_0$ and at least one element from $J_1$. %
The loop spanning lines \ref{phase2Loop1}-\ref{phase2Y0vars0} creates for each $j_0 \in J_0$, 
a set $I_{j_0}$,
which initially has exactly one unique element from $J_0$, namely $j_0$ (line \ref{phase2SetIj0}).
The loop spanning lines 
\ref{phase2Loop2}-\ref{phase2Y0vars2}
extends some of the sets $I_{j_0}$ to the sets of the second type, which are then added to $\I$ by the loop spanning lines 
\ref{phase2Loop4}-\ref{phase2ExtendPartitionI}.

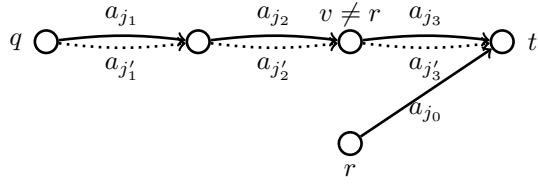
\begin{figure}[bt]
  \centering
  \def\scale{1.0}

    \scalebox{\scale}{
    \begin{tikzpicture}[
      x=1.0 cm, y=0.9 cm,
      V/.style= {circle, inner sep=3pt, draw, line width=1pt},
      P/.style= {rectangle, inner sep=1pt, draw},
      E/.style= {->, line width=1},
      L/.style  = {bend left=7pt},
      R/.style  = {bend right=7pt},
      Pointer/.style= {->, draw, opacity=0.5},
      ]
      \path (0,2.0) node [V] (d) {} +(-0.4,0) node {$q$};
      \path (2,2.0) node [V] (a) {};
      \path (4,2.0) node [V] (c) {} +(0,0.4) node {$v\ne r$};
      \path (6,2) node [V] (t) {} +(0.4,0) node {$t$};;
      \path (4,0.5) node [V] (r) {} +(0,-0.4) node {$r$};;

      \path (r) edge [E] node [below] {$a_{j_0}$}  (t);

      \path (d) edge [E,L] node [above] {$a_{j_{1}}$}  (a);
      \path (a) edge [E,L] node [above] {$a_{j_{2}}$}  (c);
      \path (c) edge [E,L] node [above] {$a_{j_{3}}$}  (t);

      \path (d) edge [E,R,dotted] node [below] {$a_{j'_{1}}$}  (a);
      \path (a) edge [E,R,dotted] node [below] {$a_{j'_{2}}$}  (c);
      \path (c) edge [E,R,dotted] node [below] {$a_{j'_{3}}$}  (t);

    \end{tikzpicture}
    }
  \caption{
    The arcs corresponding to $j_0 \in J_0$, 
    to $j_1,j_2,j_3$ found on line \ref{phase2Select1},
    and 
    to $j'_1,j'_2,j'_3$ found on line \ref{phase2Select2}
  }
  \label{figExtensionOfSecondTypeSet}
\end{figure}

For each $j_1 \in J_1 \setminus J^*$, the loop spanning lines 
\ref{phase2Loop2}-\ref{phase2Y0vars2}
finds on line \ref{phase2Select1}
a set $\{j_1,\dots,j_l\} \subseteq J \setminus \Jused$ of minimum cardinality and a $j_0\in J_0$ such that 
$W:=a_{j_1},\dots,a_{j_l}$ and $W':=a_{j_0}$ are two walks with a common node $t:=\head(a_{j_0})=\head(a_{j_l})$. 
The node $t$ have two distinct incoming arcs, \ie, $|\inarc_{W}(t) \cup \inarc_{W'}(t)| \ge 2$ (Fig. \ref{figExtensionOfSecondTypeSet}).
The reason is as follows:
Both $a_{j_0}$ and $a_{j_l}$ have the same head node $t$. 
As $j_0 \in J_0$ and $j_l \in J_l$, 
$x_{a_{j_0}}$ and $x_{a_{j_l}}$ are originated from $S_0$ and $S_l$ with $l \ne 0$, respectively. 
By Remark \ref{remark:Si:1}, $a_{j_0}$ is incident to $r$, whereas $a_{j_l}$ is not.
Thus $a_{j_0}$ and $a_{j_l}$ have different tails.
It follows that 
each found set $I:=\{j_1,\dots,j_l\}$ and element $j_0$ on 
line \ref{phase2Select1}
give a set of the second type, namely 
$\{j_0\} \cup I = \{j_0,j_1,\dots,j_l\}$.
The lines \ref{phase2ExtendSetIj0} and \ref{phase2Jused1} extend 
$I_{j_0}$ 
and $\Jused$ by $I$. %

The meaning of condition 
$y_{\head(a_{j_0}),l+1} \notin \Y_{j_0}$ 
is as follows: %
Recall that each $j \in J_i$ with $i\in \{0,\dots,H\}$ identifies the unique constraint $\S(j)=S_i^{a_j}$.
As $j_0 \in J_0$ and $j_l \in J_l$, we have $\S(j_0)=S_0^{a_{j_0}}$ and $\S(j_l)=S_l^{a_{j_l}}$.
By definitions of $S_0^{a_{j_0}}$ and $S_l^{a_{j_l}}$, 
the coefficient of variable $y_{\head(a_{j_0}),l+1}$ is 1 in constraint $S_0^{a_{j_0}}$ and $-1$ in $S_l^{a_{j_l}}$, \ie,
$c(y_{\head(a_{j_0}),l+1},S_0^{a_{j_0}})=1$ and
$c(y_{\head(a_{j_l}),l+1},S_l^{a_{j_l}})=-1$.
From $\head(a_{j_0})=\head(a_{j_l})$ follows $c(y_{\head(a_{j_0}),l+1}, \S(j_0)+\S(j_l)) = 0$. 
We add variable $y_{\head(a_{j_0}),l+1}$ to set 
$\Y_{j_0}$
in the sense that
this variable is eliminated via constraint $\S(j_0)$.

Suppose that by the end of loop \ref{phase2Loop2}, 
$I_{j_0}$ 
is extended by more than one set, say $I_1,\dots,I_m$, \ie,
$I_{j_0}
=\{j_0\} \cup I_1 \cup \dots \cup I_m$.
We have seen that for each $I \in I_1,\dots,I_m$, set $\{j_0\} \cup I$ is a set of the second type.
To make sure that it is also true for 
$I_{j_0}=\{j_0\} \cup I_1 \cup \dots \cup I_m$,
we must also ensure for each two $I \ne I' \in \{I_1,\dots,I_m\}$ that $I \cup I'$ is a set of the second type.
To guarantee this,
the loop on line \ref{phase2Loop3} finds the sets 
$\{j'_1,\dots,j'_l\}$ 
such that
the walks $W:=a_{j_1},\dots,a_{j_l}$ and $W'=a_{j'_1},\dots,a_{j'_l}$ are the same (Fig. \ref{figExtensionOfSecondTypeSet}) 
and thus have no common node $v$ with $|\inarc_W(v) \cup \inarc_{W'}(v)| \ge 2$. 
Hence the set $\{j_1,\dots,j_l\} \cup \{j'_1,\dots,j'_l\}$ is not of the second type.
So we add $j'_1,\dots,j'_l$ to $\Jused$ on line \ref{phase2Jused2} to prevent them from becoming part of $I_{j_0}$.
Phase 3 (line \ref{phase3ExtendPartitionIBySetType3}) adds the remaining set $J \setminus \bigcup_{I \in \I} I$ to $\I$ as the only set of the third type.
\\

\def\disjoint{1}
\def\settypes{2}
\def\settypesI{2a}
\def\settypesII{2b}

\noindent 
We now show the correctness of Algorithm \ref{algoPartitionI}, \ie, the algorithm constructs a set-family $\I$ such that
\begin{description}
\item[(\disjoint)] $\I$ is a partition of $J$.
\item[(\settypes)] %
  Each set in $\I$ belongs to one of the three type of the sets.
\end{description}

\noindent \textbf{\emph{(\disjoint)}} 
To this end, we show that $\bigcup_{I \in \I} I = J$ such that each two distinct sets $I \ne I' \in \I$ satisfy $I \cap I' = \{\}$.
Since any element in $\I$ is a subset of $J$, we have $\bigcup_{I \in \I} I \subseteq J$.
On the other hand, Phase 3 (line \ref{phase3ExtendPartitionIBySetType3}) ensures that for each element $j \in J$, there is a set $I \in \I$ with $j \in I$, so $J \subseteq \bigcup_{I \in \I} I$ and thus $\bigcup_{I \in \I} I = J$.
The algorithm manages the set $\Jused$ of used elements in $J \setminus J_0$.
As soon as an element is selected from $J \setminus J_0$, it is added into the set $\Jused$ and into at most one set of partition $\I$.
The next element is selected from set $J \setminus \Jused$. 
So each element of $J \setminus J_0$ is included at most one set $I \in \I$.
The elements of $J_0$ are contained only in the sets of the second and third types.
There is no element $j_0 \in J_0$ and no two sets $I \ne I'$ of the second or third type such that $j_0 \in I$ and  $j_0 \in I'$.
Thus each element of $J$ is included at most one set $I \in \I$. Hence each two distinct sets $I \ne I' \in \I$ satisfy $I \cap I' = \{\}$.
\\

\noindent \textbf{\emph{(\settypes)}} 
To this end, we show that each element of $J_1$ is in a set $I \in \I$ of the first or second type. 
There is only one set of the third type, which does not contain any element of $J_1$.
The first phase finds for each $j_1 \in J_1 \cap J^*$ a set of the first type, such that after this phase, 
no $j_1 \in J_1 \setminus J^*$ meets the conditions in line \ref{phase1Select}, and thus
no subset of $J \setminus J^*$ is a set of the first type.
So we need to show that the second phase finds a set of the second type for each element in $J_1 \setminus J^*$.
More precisely, we show the following:

\begin{description}
  \item[(\settypesI)] The second phase finds for each element $j_1 \in J_1 \setminus J^*$ a set $\{j_1,\dots,j_l\}$ and an element $j_0 \in J_0$ that meet the conditions of line \ref{phase2Select1}.

  \item[(\settypesII)] Let $I_1,\dots,I_m$ be the sets that are joined to $I_{j_0}$ on 
  line \ref{phase2ExtendSetIj0}, \ie, $I_{j_0} =\{j_0\} \cup I_1 \cup \dots \cup I_m$.
  Then for each two $I \ne I' \subseteq I_{j_0}$, set $I \cup I'$ is a set of the second type.
\end{description}

\noindent \textbf{\emph{(\settypesI)}} 
To show this, %
we need some facts about the coefficients $c(y_{v,i},C)$ of the variables $y_{v,i}$ in some crucial constraints $C$.
Let $\Jround 0 := J^*$ (line \ref{phase2Jstar}) 
and let $\Jround k \subseteq J$ with $k\ge 1$ be the %
set of elements added to $\Jused$ at iteration $k$ 
of the loop spanning lines \ref{phase2Loop2}--\ref{phase2Y0vars2}, 
and let 
$\Jround {\le k}:= \bigcup_{0\le i \le k} \Jround {i}$. 

Let $j_1 \in J_1$ and let $I:=\{j_1,\dots,j_H\} \subseteq \Jround 0$ 
be a set found in the first phase on line \ref{phase1Select}. 
Due to the condition in line \ref{phase1Select}, we have $(j_{1},\dots,j_{H}) \in J_1 \times \dots \times J_H$.
Recall that as $j_i \in J_i$, we have $\S(j_i)=S_i^{a_{j_i}}$,
so $(\S(j_{1}),\ldots,\S(j_{H}))=(S_1^{a_{j_{1}}},\ldots,S_{H}^{a_{j_{H}}})$.
Since $a_{j_1},\dots,a_{j_H}$ is a walk, %
the sum of constraints $S_1^{a_{j_{1}}},\ldots,S_{H}^{a_{j_H}}$, \ie,
$\sum_{j \in I} \S(j)$
has the following form:
\begin{align*}
  \msum_{j \in I} \S(j):
  & & 
  \big(- \msum_{i=2}^H y_{\tail(a_{j_1}),i} \big)	+ (x_{a_{j_1}}+\dots+x_{a_{j_H}}) \le H-1.
\end{align*}
Hence for any \yvar\ $y_{v,i}$, we have 
$c \big(y_{v,i}, \sum_{j \in I} \S(j) \big) \le 0$.
As this holds for any set $I \subseteq \Jround 0$ that is found on line \ref{phase1Select},
it follows:
\begin{align} \label{constraintwalk:cJ0Yvh}
  c \big(y_{v,i}, \msum_{j \in \Jround 0} \S(j) \big) \le 0
  && \text{ for any \yvar\ } y_{v,i}.
\end{align}

Next we study the coefficient of $y_{v,i}$ in the constraint $\sum_{j \in \Jround k} \S(j)$ for a $k\ge1$.
Let $j_1 \in J_1$ and let 
$\{j_{1},\dots,j_{l}\}$ 
with 
$l<H$ and
$(j_{1},\dots,j_{l}) \in J_1 \times \dots \times J_{l}$
be the set joined to $\Jround k$
on line \ref{phase2Jused1}.
The constraints corresponding to $(j_1,\dots,j_l)$ are $(\S(j_1),\ldots,\S(j_l))$.
From 
$(j_{1},\dots,j_{l}) \in J_1 \times \dots \times J_{l}$
follows
$(\S(j_1),\ldots,\S(j_l))=(S_1^{a_{j_1}},\ldots,S_l^{a_{j_l}})$, \ie,
the constraints $\sum_{j \in \{j_1,\dots,j_l\} } \S(j)$ and 
$S_1^{a_{j_1}}+\dots+S_l^{a_{j_l}}$ are the same.
Let $q:=\tail(a_{j_1})$ and $t:=\head(a_{j_l})$ (Fig. \ref{figExtensionOfSecondTypeSet}).
The constraint
$S_1^{a_{j_1}}+\dots+S_{l}^{a_{j_l}}$
with $l<H$ and a walk $a_{j_1},\dots,a_{j_l}$
has the form:
\begin{align*}
    S_1^{a_{j_1}}+\dots+S_l^{a_{j_l}}:
    &&
    (- \msum_{h=2}^H y_{q,h} )	
    & + (x_{a_{j_1}}+\dots+x_{a_{j_l}}) 
    - y_{t,l+1} \le l-1.
\end{align*}
Thus we have:
\begin{align}
  &c \big( y_{q,h}, \msum_{j \in \{j_1,\dots,j_l\} } \S(j) \big) =
  c \big( y_{q,h},  S_1^{a_{j_1}}+\dots+S_l^{a_{j_l}} \big) \le -1
  && \mbox{for any } h\ge 2.
  \label{constraintwalk:tail}
  \\
  &c \big( y_{v,h}, \msum_{j \in \{j_1,\dots,j_l\} } \S(j) \big) =
  c \big( y_{v,h}, S_1^{a_{j_1}}+\dots+S_l^{a_{j_l}} \big) \le 0
  && \mbox{for any \yvar\ $y_{v,h}$ }
  \\
  &c \big( y_{t,l+1}, \msum_{j \in \{j_1,\dots,j_l\} } \S(j) \big) =
  c \big( y_{t,l+1}, S_1^{a_{j_1}}+\dots+S_l^{a_{j_l}} \big) \le -1.
  \label{constraintwalk:cS1b}
\end{align}
Let $\{j'_1,\dots,j'_l\} \subseteq J$ 
with $(j'_1,\dots,j'_l) \in J_i \times \dots \times J_{i+l-1}$
and $2 \le i \le H-l$ 
and
$(a_{j'_1},\dots,a_{j'_l}) = (a_{j_1},\dots,a_{j_l})$
be a set joined to $\Jround k$
on line \ref{phase2Jused2}.
From $(j'_1,\dots,j'_l) \in J_i \times \dots \times J_{i+l-1}$ follows
$(\S(j'_1),\ldots,\S(j'_l))=(S_i^{a_{j'_1}},\ldots,S_{i+l-1}^{a_{j'_l}})$.
Since $(a_{j'_1},\dots,a_{j'_l}) = (a_{j_1},\dots,a_{j_l})$, 
we have
$(S_i^{a_{j'_1}},\ldots,S_{i+l-1}^{a_{j'_l}})= (S_i^{a_{j_1}},\ldots,S_{i+l-1}^{a_{j_l}})$, so
$(\S(j'_1),\ldots,\S(j'_l))= (S_i^{a_{j_1}},\ldots,S_{i+l-1}^{a_{j_l}})$.
Thus the constraints $\sum_{j \in \{j'_1,\dots,j'_l\} } \S(j)$ and 
$S_i^{a_{j_1}}+\dots+S_{i+l-1}^{a_{j_l}}$
are the same.
The constraint
$S_i^{a_{j_1}} + \dots + S_{i+l-1}^{a_{j_l}}$
with $2 \le i \le H-l$ and
a walk $a_{j_1},\dots,a_{j_l}$ 
has the form:
\begin{align*}
  S_i^{a_{j_1}} + \dots + S_{i+l-1}^{a_{j_l}}:
  &&
    y_{q,i}			
    & + (x_{a_{j_1}}+\dots+x_{a_{j_l}}) 
    - y_{t,l+i} \le l.
    \notag
\end{align*}
Thus, for $2 \le i \le H-l$, we have
\begin{align}
  &c ( y_{q,i}, \allowbreak  S_i^{a_{j_1}} + \dots + S_{i+l-1}^{a_{j_l}}) = 1 
  \label{constraintwalkIge2Yqh}
  \\
  &c \big( y_{v,h},  S_i^{a_{j_1}} + \dots + S_{i+l-1}^{a_{j_l}} \big) \le 0
  & \mbox{for any \yvar\ $y_{v,h}$ except $y_{q,i}$ }
  \\
  &c \big( y_{t,l+i}, S_i^{a_{j_1}} + \dots + S_{i+l-1}^{a_{j_l}}\big) \le -1.
  \label{constraintwalk:Ige2:ytl}
\end{align}

Let $I \subseteq \{2,\dots,H-l\}$ such that 
for each $i \in I$ at iteration $i$ of the loop spanning lines \ref{phase2Loop3}--\ref{phase2Y0vars2}
the elements $(j'_1,\dots,j'_l) \in J_i \times \dots \times J_{i+l-1}$
are joined to $\Jround k$
on line \ref{phase2Jused2}.
Note that $(S_1^{a_{j_1}},\dots,S_l^{a_{j_l}})$ can be written as 
$(S_i^{a_{j_1}} , \dots , S_{i+l-1}^{a_{j_l}})$ where $i=1$.
Thus
$\sum_{j \in \Jround k} \S(j)$ 
is the same constraint as
$\sum_{i \in \{1\} \cup I } ( S_i^{a_{j_1}} + \dots + S_{i+l-1}^{a_{j_l}} ) $.
Facts (\ref{constraintwalk:tail}) and (\ref{constraintwalkIge2Yqh}) imply
$c(y_{q,i}, \sum_{i \in \{1\} \cup \{i\} } ( S_i^{a_{j_1}} + \dots + S_{i+l-1}^{a_{j_l}} )) \le 0$ for any $2 \le i \le H-l$.
Particularly, (\ref{constraintwalk:tail})--(\ref{constraintwalk:Ige2:ytl}) imply that
$c(y_{v,h}, \sum_{i \in \{1\} \cup I } ( S_i^{a_{j_1}} + \dots + S_{i+l-1}^{a_{j_l}} )) \le 0$ for any \yvar\  $y_{v,h}$. 
It follows:

\begin{align*}
  c \big( y_{v,h}, \msum_{j \in \Jround {k}} \S(j) \big) \le 0 
  && \text{ for any $k\ge1$ and for any \yvar\ } y_{v,h}.
\end{align*}
This fact, together with (\ref{constraintwalk:cJ0Yvh}), implies that:
\begin{align}
  c \big( y_{v,h}, \msum_{j \in \Jround {k}} \S(j) \big) \le 0 
  && \text{ for any $k\ge0$ and for any \yvar\ } y_{v,h}.
  \label{constraintwalk:cJkYvh}
\end{align}
Particularly,
by (\ref{constraintwalk:cS1b}) and (\ref{constraintwalk:Ige2:ytl}),
$c \big( y_{t,l+i}, S_i^{a_{j_1}} + \dots + S_{i+l-1}^{a_{j_l}}\big) \le -1$
holds for any $i\ge 1$.
Then we have
$c(y_{t,l+1}, \sum_{i \in \{1\} \cup I } ( S_i^{a_{j_1}} + \dots + S_{i+l-1}^{a_{j_l}} )) \le -1$.
It follows that:
\begin{align}
  c \big( y_{t,l+1}, \msum_{j \in \Jround {k}} \S(j) \big) \le -1
  && \text{ for any } k\ge1.
\label{constraintwalk:cJkYtl}
\end{align}
\noindent
Using these facts about the coefficients of the $y$ variables 
we show \emph{(\settypesI)}. %
Suppose, for contradiction, $k$ is the first iteration
of the loop spanning lines \ref{phase2Loop2}--\ref{phase2Y0vars2}, 
at which there is no $j_0 \in J_0$ and $(j_1,\dots,j_l) \in J_1 \times \dots \times J_l$ that meet the conditions of line \ref{phase2Select1}.
Recall that after the first phase, there is no set of the first type, 
\ie,
there is no $\{j_1,\dots,j_l\} \in J \setminus J^*$ with $l=H$ (but with $l \le H-1$),
such that 
$(j_1,\dots,j_l) \in J_1 \times \dots \times J_l$ 
and 
$(a_{j_1},\dots,a_{j_l})$ is a walk.
Consider the subset 
$\{j_1,\dots,j_l\} \subseteq J \setminus \Jround{\le k-1}$ %
such that 
$(j_1,\dots,j_l) \in J_1 \times \dots \times J_l$,
and 
$a_{j_1},\dots,a_{j_l}$ is a walk, 
and there is no 
$j_{l+1} \in 
 J_{l+1} \setminus \Jround{\le k-1}$
with $\head (a_{j_l})=\tail(a_{j_{l+1}})$.
Let $t:=\head(a_{j_l})$.
We have 
$c(y_{t,l+1}, \sum_{j \in \{j_1,\dots,j_l\}} \S(j)) \le -1$
by (\ref{constraintwalk:cS1b})
and 
$c(y_{t,l+1}, \sum_{j \in J^{(\le k-1)}} \S(j)) \le 0$
by (\ref{constraintwalk:cJkYvh}).
It follows
$c(y_{t,l+1}, \sum_{j \in J^{(\le k-1)} \cup \{j_1,\dots,j_l\}} \S(j)) \le -1$.
Assume
\begin{align}
c \big ( 
    y_{t,l+1}, 
    \msum_{j \in J^{(\le k-1)}} \S(j)
  \big) &=-z_1 
  && z_1 \in \mathbb N_0, 
  \label{constraintwalkCoeffs1}
  \\
c \big( 
    y_{t,l+1}, 
    \msum_{j \in J^{(\le k-1)} \cup \{j_1,\dots,j_l\}} \S(j)
  \big) &= -z_2 
  &&
  z_2 \in \mathbb N,\  z_2 > z_1.
  \label{constraintwalkCoeffs2}
\end{align}
Recall that $X$ is a unweighted sum $C_1 + \dots + C_l$ of not necessarily distinct constraints $C_1,\dots,C_l$ selected from sets $S_0,\dots,S_H$ and (\ref{Q2}).
By definitions of $C_i$ with $i \in \{1,\dots,l\}$, we have $c(y_{v,h}, C_i ) \in \{-1,0,1\}$ for each $y$-variable $y_{v,h}$.
Since $X$, \ie, $\sum_{i=1}^l C_i$ is an $x$-space constraint, we have $c(y_{v,h},\sum_{i=1}^l C_i) = 0$.
It follows that if there is $i \in \{1,\dots,l\}$ with $c(y_{v,h},C_i) = -1$, then there is $i' \in \{1,\dots,l\} \setminus \{i\}$ 
with $c(y_{v,h},C_{i'}) = 1$.
Since $y_{v,h}$ occurs with a positive sign only in $S_0^{a}$ with $\head(a)=v$ and in $S_h^{a'}$ with $\tail(a')=v$, 
we have $C_{i'} \in \{ S_0^{a} \colon \head(a)=v \} \cup \{ S_h^{a'} \colon \tail(a')=v\}$, 
\ie, $C_{i'} \in S_0 \cup  S_h$.
Particularly, if there is $I \subset \{1,\dots,l\}$, such that $c(y_{v,h}, \sum_{i \in I} C_i  ) = -z$ for a $z>0$, 
then there is $I' \subset \{1,\dots,l\} \setminus I$ such that 
$c(y_{v,h}, \sum_{i' \in I'} C_{i'} ) = z$, 
where 
$C_{i'} \in S_0 \cup  S_h$ for each $i' \in I'$.
Thus, as there is $I:=\Jround{\le k-1} \cup \{j_1,\dots,j_l\}$ such that 
$c( y_{t,l+1}, \sum_{j \in I} \S(j)) = -z_2$ 
by (\ref{constraintwalkCoeffs2}),
there is also $I' \subseteq J \setminus I$
with
$c( y_{t,l+1}, \sum_{j \in I'} \S(j)) \ge z_2$,
where 
$\S(j) \in \{ S_0^{a} \colon \head(a)=t \} \cup \{ S_{l+1}^{a} \colon \tail(a)=t\}$ for each $j \in I'$.
From $\S(j) \in S_0 \cup  S_{l+1}$ follows $I' \in J_0 \cup J_{l+1}$.
Moreover, $I' \subseteq J \setminus I$ implies $I' \in (J_0 \setminus I) \cup (J_{l+1} \setminus I)$.
In other words, to eliminate $y_{t,l+1}$ there must be a 
$j_0 \in J_0 \setminus I =J_0$
with 
$\head(a_{j_0})= t$,
or
$j_{l+1} \in J_{l+1} \setminus I $
with $\tail(a_{j_{l+1}})=t$. %
By construction of $\{j_1,\dots,j_l\}$ there is no such $j_{l+1}$.
It follows that 
\begin{align}
c \big( y_{t,l+1}, \msum_{j \in J_0} \S(j) \big) \ge z_2,
\label{constraintwalkCoeffsJ0}
\end{align}
\ie, there is a $j_0 \in J_0 $ with $\head(a_{j_0})= t = \head(a_{j_l})$. 

As the last condition, line \ref{phase2Select1} requires a $j_0$, such that 
$y_{t,l+1} \notin \Y_{j_0}$.
If there is such $j_0$ at iteration $k$, then we are done, so suppose not.
Then $|\{ j_0 \in J_0 \colon y_{t,l+1} \in \Y_{j_0} \}| 
= c( y_{t,l+1}, \sum_{j \in J_0} \S(j))$ 
before iteration $k$.
From (\ref{constraintwalkCoeffsJ0}) follows that 
$|\{ j_0 \in J_0 \colon y_{t,l+1} \in \Y_{j_0} \}| \ge z_2$ 
before iteration $k$. %
Consider an iteration $k'$ of the loop spanning lines \ref{phase2Loop2}--\ref{phase2Y0vars2}.
The algorithm adds $y_{t,l+1}$ to a 
$\Y_{j_0}$ 
on line \ref{phase2Y0vars1} 
or 
on line \ref{phase2Y0vars2}. 
In case of line \ref{phase2Y0vars1} 
the algorithm also adds on line \ref{phase2Jused1}
the elements $j_1,\dots,j_l$ 
to $\Jround {k'}$ 
such that
$(j_1,\dots,j_{l}) \in J_1 \times \dots \times J_{l}$
and
$a_{j_1},\dots,a_{j_l}$ build a walk with $l \le H-1$ arcs
and
$\head(a_{j_l})=t$. %
In case of line \ref{phase2Y0vars2} 
the algorithm also adds on line \ref{phase2Jused2}
the elements $j_i,\dots,j_l$ with $i\ge 2$ to $\Jround {k'}$ 
such that
$(j_i,\dots,j_l) \in J_i \times \dots \times J_l$ %
and
$a_{j_i},\dots,a_{j_l}$ build a walk with $l \le H-1$ arcs
and
$\head(a_{j_l})=t$. %
It follows that if the algorithm adds $y_{t,l+1}$ to a 
$\Y_{j_0}$ then it also adds
the elements $j_i,\dots,j_l$ with $i\ge 1$ to $\Jround {k'}$ 
such that
$(j_i,\dots,j_l) \in J_i \times \dots \times J_l$
and
$a_{j_i},\dots,a_{j_l}$ build a walk with $l \le H-1$ arcs
and
$\head(a_{j_l})=t$. 
Then $c( y_{t,l+1}, \sum_{j\in \Jround {k'}} \S(j)) \le -1$ by (\ref{constraintwalk:cJkYtl}).
It follows 
$c(y_{t,l+1}, \sum_{j \in J^{(\le k-1)}} \S(j)) 
\le - |\{ j_0 \in J_0 \colon y_{t,l+1} \in \Y_{j_0} \}| \le - z_2 < - z_1$.
It contradicts 
$c( y_{t,l+1}, \sum_{j \in J^{(\le k-1)}} \S(j) ) =-z_1$. %

\noindent \textbf{\emph{(\settypesII)}} 
Let $I:=\{j_1,\dots,j_l\} \in \{I_1,\dots,I_m\}$.
By line \ref{phase2Select1}, 
$W:=a_{j_1},\dots,a_{j_l}$
and
$W':=a_{j_0}$
are two walks with a common node $t:=\head(a_{j_0})=\head(a_{j_l})$. 
The node $t$ have two distinct incoming arcs, \ie, $|\inarc_{W}(t) \cup \inarc_{W'}(t)| \ge 2$ (Fig. \ref{figExtensionOfSecondTypeSet}).
The reason is as follows:
Both $a_{j_0}$ and $a_{j_l}$ have the same head node $t$. 
As $j_0 \in J_0$ and $j_l \in J_l$ with $l \ne 0$, 
$x_{a_{j_0}}$ and $x_{a_{j_l}}$ are originated from $S_0$ and $S_l$, respectively. 
By Remark \ref{remark:Si:1}, $a_{j_0}$ is incident to $r$, whereas $a_{j_l}$ with $l \ne0$ is not.
Thus $a_{j_0}$ and $a_{j_l}$ have different tails.
It follows that for each $I \in I_1,\dots,I_m$ the union of the sets $\{j_0\}$ and $I$ is a set of the second type.

\begin{figure}[bt]
  \centering
  \def\scale{0.99}
    \scalebox{\scale}{
    \begin{tikzpicture}[
      x=1.0 cm, y=0.9 cm,
      V/.style= {circle, inner sep=3pt, draw, line width=1pt},
      P/.style= {rectangle, inner sep=1pt, draw},
      E/.style= {->, line width=1},
      L/.style  = {bend left=7pt},
      R/.style  = {bend right=7pt},
      Pointer/.style= {->, draw, opacity=0.5},
      ]
      \path (-2,2.0) node [V] (b) {}; %
      \path (0,2.0) node [V] (d) {};%
      \path (2,2.0) node [V] (a) {};
      \path (4,2.0) node [V] (c) {}; %
      \path (6,2) node [V] (t) {} +(0.4,0) node {$t$};;

      \path (b) edge [E]   node [above] {$a_{j_{1}}$}  (d);
      \path (d) edge [E,L] node [above] {$a_{j_{2}}$}  (a);
      \path (a) edge [E,L] node [above] {$a_{j_{3}}$}  (c);
      \path (c) edge [E,L] node [above] {$a_{j_{4}}$}  (t);

      \path (d) edge [E,R,dotted] node [below] {$a_{j'_{1}}$}  (a);
      \path (a) edge [E,R,dotted] node [below] {$a_{j'_{2}}$}  (c);
      \path (c) edge [E,R,dotted] node [below] {$a_{j'_{3}}$}  (t);

    \end{tikzpicture}
    }
  \caption{
    The walk $W=a_{j_1},a_{j_2},a_{j_3},a_{j_4}$ and a suffix $W'=a_{j'_1},a_{j'_2},a_{j'_3}$ of $W$.
  }
  \label{figExtensionOfSecondTypeSet2}
\end{figure}
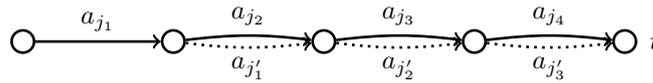

We now need only to prove for each two $I \ne I' \in \{I_1,\dots,I_m\}$ that $I \cup I'$ is a set of the second type.
Suppose, for contradiction, $I \cup I'$ is not a set of second type. 
Let $I:=\{j_1,\dots,j_l\}$
and
$I':=\{j'_1,\dots,j'_q\}$.
By line \ref{phase2Select1}, 
$W:=a_{j_1},\dots,a_{j_l}$ and $W':=a_{j'_1},\dots,a_{j'_q}$
are two walks with a same end node $t:=\head(a_{j_0})$, \ie, 
$\head(a_{j_l})=\head(a_{j'_q})=t$. 
Since $I \cup I'$ is not a set of the second type, $W$ and $W'$ does not have any common node $v$ with two distinct incoming arcs, \ie, with $|\inarc_{W}(v) \cup \inarc_{W'}(v)| \ge 2$ (Fig. \ref{figExtensionOfSecondTypeSet2}).
Then $W$ is a suffix of $W'$ or vice versa, \ie, $W \sqsupset W'$ or $W' \sqsupset W$.
Assume that $I$ and $I'$ are added to $I_{j_0}$ at iterations $k$ and $k'<k$ of the loop spanning lines \ref{phase2Loop2}--\ref{phase2Y0vars2}, respectively.
By the condition of line \ref{phase2Select1}, 
the second phase finds a set of minimum cardinality first, so we have $|I'| \le |I|$. It follows that $\len{W'} \le \len{W}$ and thus $W' \sqsupset W$.
Then there is an $i \in \{1,\dots,l\}$ such that 
$a_{j'_1},\dots,a_{j'_q}=a_{j_i},\dots,a_{j_l}$.
Since $I=\{j_1,\dots,j_l\}$ is added to $I_{j_0}$ at iteration $k$, 
by the condition of line \ref{phase2Select1}, 
at the start of iteration $k$, 
we have $y_{t,l+1} \notin \Y_{j_0}$.
Then this is true for the iteration $k' < k$, too.

There are two cases: 
If $\len{W'} = \len{W}$, 
then
$a_{j'_1},\dots,a_{j'_q}=a_{j_1},\dots,a_{j_l}$, \ie, $i=1$. %
Then at iteration $k' < k$,
in line \ref{phase2Y0vars1}, 
$y_{t,l+1}$ is added to $\Y_{j_0}$, so
after iteration $k'$, thus at the start of iteration $k>k'$,
we have $y_{t,l+1} \in \Y_{j_0}$, a contradiction.

Otherwise $\len{W'} < \len{W}$. 
Then $a_{j'_1},\dots,a_{j'_q} = a_{j_i},\dots,a_{j_l}$ with $i \ge 2$.
Recall that $I \subseteq \Jround {k}\subseteq J \setminus J^{(\le k-1)}$, 
while $I'\subseteq \Jround {k'}\subseteq J \setminus J^{(\le k'-1)}$.
Since $k' < k$, we have 
$J \setminus J^{(\le k-1)} 
\subseteq 
J \setminus J^{(\le k'-1)}$,
and thus
$I \subseteq J \setminus J^{(\le k'-1)}$.
Then the subset $\{j_i,\dots,j_l\}$ of $I=\{j_1,\dots,j_l\}$ is a subset of $J \setminus J^{(\le k'-1)}$, too.
Thus $\{j_i,\dots,j_l\}$ was present in iteration $k'$ as $I'=\{j'_1,\dots,j'_q\}$ was.
Then $\{j_i,\dots,j_l\}$ would satisfy 
at iteration $k' < k$ of loop \ref{phase2Loop2}, 
at iteration $i$ of loop \ref{phase2Loop3},
the conditions in line \ref{phase2Select2},
\ie, 
$(j_i,\dots,j_l) \in J_i \times \dots \times J_l$ and 
$a_{j_i},\dots,a_{j_l} = a_{j'_1},\dots,a_{j'_q}$.
Then the variable
$y_{t,l+1}$ would be added to 
$\Y_{j_0}$
at iteration $k' < k$, in line \ref{phase2Y0vars2},
so after iteration $k'$, thus at the start of iteration $k>k'$,
$y_{t,l+1} \in \Y_{j_0}$, a contradiction.
\fi
\end{proof}

\subsubsection{Results for the \hstp}
\begin{theorem} \label{theorem:hstp}
  For the \hstp, the partial-ordering model is strictly stronger than the assignment model.
\end{theorem}
\begin{proof}
From %
definitions of %
$\passHSTP$, $\ppopHSTP$
follows %
\ifshort
$\xassHSTP$$=$$ \xassHT \cap \{ x \colon x \mbox{ meets } (\ref{constr:indeg:ge1}) \}$, 
\else
$\xassHSTP = \xassHT \cap \{ x \colon x \mbox{ meets } (\ref{constr:indeg:ge1}) \}$, 
\fi
$\xpopHSTP = \xpopHT \cap \{ x \colon x \mbox{ meets } (\ref{constr:indeg:ge1}) \}$.
Theorem \ref{theorem:hstp:fme} implies 
$\xpopHSTP \subseteq \xassHSTP$.
Let $(x^*,l^*,g^*)$ be an optimal solution, with objective value $\val^*$, %
of the LP relaxation of (\ref{Phstp}).
From $\xpopHSTP \subseteq \xassHSTP$ follows $x^* \in \xassHSTP$.
Then, since both models have the same objective function, which depends only on $x$ variables,
the LP relaxation of (\ref{Ahstp}) also has a feasible solution (\eg, $x^*$) whose objective value is not greater than $\val^*$.
Hence $\val_{\text{\ref{Phstp}}} \ge \val_{\text{\ref{Ahstp}}}$ holds for all \hstp instances,
so (\ref{Phstp}) is stronger than (\ref{Ahstp}).
Moreover, the computations (Sec. \ref{sec:evaluation}) have detected the \hstp instances  for which the inequality is strict, so (\ref{Ahstp}) is not stronger than (\ref{Phstp}).
\end{proof}

\subsubsection{Results for the \hmst} %
\begin{theorem} \label{theorem:hmst}
  For the \hmst, the partial-ordering model is strictly stronger than the assignment model.
\end{theorem}
\begin{proof}
The \hmst\ is a special case of the \hstp with $R = V$.
Hence the assignment and partial-ordering based models for the \hmst also take the form 
(\ref{Ahstp}) and (\ref{Phstp}), respectively.
In Theorem \ref{theorem:hstp}, we have shown that $\val_{\text{\ref{Phstp}}} \ge \val_{\text{\ref{Ahstp}}}$. 
Moreover, there are the \hmst instances (\autoref{sec:evaluation}) for which the inequality is strict.
\end{proof}

\subsubsection{Results for the \stprbh}
The solution $T$ of the \stprbh is a HT, such that
$\sum_{e \in E(T)} c_e \le B$, formulated~\cite{DBLP:conf/alenex/JabrayilovM19} as:
\begin{align} \label{constr:budget}
  \msum_{uv \in E} c_{u,v} (x_{u,v} + x_{v,u})  \le B.
\end{align}
Let $\passSTPRBH$ and $\ppopSTPRBH$ 
denote the polytopes of the assignment and
partial-ordering based LPs for the \stprbh, respectively. Then we have:
\begin{align*}
\passSTPRBH = & \{ (x,y) \colon (x,y) \in \passHT\ (\ref{pass1});\ x \mbox{ satisfies } (\ref{constr:budget}) \},\\
\ppopSTPRBH = & \{ (x,l,g) \colon (x,l,g) \in \ppopHT\ (\ref{ppop});\ x \mbox{ satisfies } (\ref{constr:budget})\}.
\end{align*}
The goal of the \stprbh is to maximize $\rho_{r} + \sum_{uv \in E} (x_{u,v}\rho_v  + x_{v,u}\rho_u )$ \cite{DBLP:conf/alenex/JabrayilovM19},
\ie, its objective function also depends only on $x$ variables:
\ifshort
\begin{align}
  \max & \big\{ 
    \rho_{r} + \msum_{uv \in E} (x_{u,v}\rho_v  + x_{v,u}\rho_u )
    \colon
    (x,y) \in \passSTPRBH \text{ int.}
    \big\}
    \label{Astprbh}
    \tag{A-\stprbh}
    \\
  \max & \big\{ 
    \rho_{r} + \msum_{uv \in E} (x_{u,v}\rho_v  + x_{v,u}\rho_u )
    \colon
    (x,l,g) \in \ppopSTPRBH \text{ int.}
    \big\}
    \label{Pstprbh}
    \tag{P-\stprbh}
\end{align}
\else
\begin{align}
  \max & \Big\{ 
    \rho_{r} + \msum_{uv \in E} (x_{u,v}\rho_v  + x_{v,u}\rho_u )
    \colon
    (x,y) \in \passSTPRBH;\  x,y \text{ integral}
    \Big\}
    \label{Astprbh}
    \tag{A-\stprbh}
    \\
  \max & \Big\{ 
    \rho_{r} + \msum_{uv \in E} (x_{u,v}\rho_v  + x_{v,u}\rho_u )
    \colon
    (x,l,g) \in \ppopSTPRBH;\  x,l,g \text{ integr.}
    \Big\}
    \label{Pstprbh}
    \tag{P-\stprbh}
\end{align}
\fi

\begin{theorem} \label{theorem:stprbh}
  For the \stprbh, the partial-ordering model is strictly stronger than the assignment model.
\end{theorem}
\begin{proof}
From definitions of 
$\passSTPRBH$, $\ppopSTPRBH$ follows
$\xassSTPRBH = \xassHT \cap \{ x \colon \allowbreak x \mbox{ satisfies } (\ref{constr:budget}) \}$, 
$\xpopSTPRBH = \xpopHT \cap \{ x \colon x \mbox{ satisfies } (\ref{constr:budget}) \}$.
Theorem \ref{theorem:hstp:fme} implies 
$ \xpopSTPRBH \subseteq \xassSTPRBH $.
Analogous to Theorem \ref{theorem:hstp}, we can show that 
$\val_{\text{\ref{Pstprbh}}} \le \val_{\text{\ref{Astprbh}}}$ holds for 
all \stprbh instances.
Moreover, there are the \stprbh instances for which the inequality is strict \cite{DBLP:conf/alenex/JabrayilovM19}.
\end{proof}

\ifshort
\else
\noindent 
Theorem \ref{theorem:hstp:fme} allows us to summarize the polyhedral results as follows: 
\begin{corollary} \label{corollaryHT}
For the problems whose solution is a hop-constrained tree, a rooted tree that has bounded depth, the partial-ordering model is stronger than the assignment model if both models have the same objective function and constraints (\ie, depending only on common $x$-variables), except the hop-constrained tree constraints, namely $(x,y) \in \passHT$~(\ref{pass1}) and  $(x,l,g) \in \ppopHT$~(\ref{ppop}).
\hfill $\blacktriangleleft$
\end{corollary}
\fi
\section{Computational comparison}
\label{sec:evaluation}

\ifshort %

\else %
\def\scale{0.62}
\def\scale{0.72}

\begin{figure*}[b]
  \captionsetup[subfigure]{justification=centering}
  \begin{subfigure}[b] {0.35\textwidth}
  \begin{tikzpicture}
    \path (0, 0.0)    node {\includegraphics[scale=\scale]{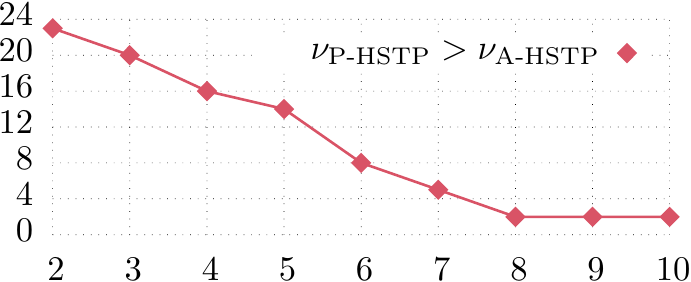}};
    \path ( 0,-1.05)    node {};
  \end{tikzpicture}
	\vspace{3pt}
	\subcaption{}
	\label{plotNumberStrong}
  \end{subfigure}
  \centering
  \begin{subfigure}[b] {0.64\textwidth}
  \begin{tikzpicture}
    \path ( 0,0)    node {\includegraphics[scale=\scale]{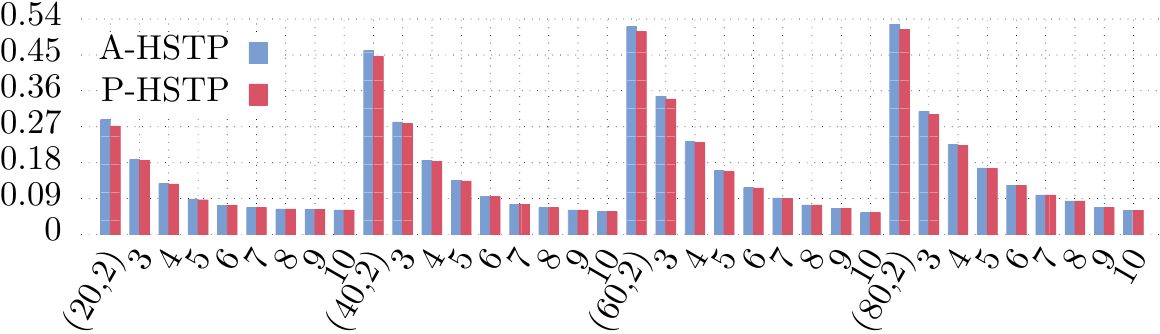}};
  \end{tikzpicture}
	\vspace{-7pt}
	\subcaption{}
	\label{plotMeangap}
  \end{subfigure}
  \vspace{-20pt}
  \caption{
    (\subref{plotNumberStrong}) Number of instances 
    with strict inequality $\val_{\text{\ref{Phstp}}} > \val_{\text{\ref{Ahstp}}}$ 
    depending on $H$;
    (\subref{plotMeangap}) Average LP gaps of (\ref{Ahstp}) and (\ref{Phstp}) over 216 instances
    depending on $(|V|,H)$
  }
  \label{plotGap}
\end{figure*}
\fi %

The assignment and partial-ordering models give (up to our knowledge) the best two state-of-the-art formulations
\cite{DBLP:conf/alenex/JabrayilovM19,DBLP:journals/mpc/SinnlL16}
for the \stprbh.
Both models solve almost all 414 DIMACS instances \cite{DIMACS11BENCHMARK} with up to 500 nodes and \numprint{12500} edges, whereas previous models left 86 unsolved. 
A computational comparison of the two models for the \stprbh, using the DIMACS instances, have already been done in \cite{DBLP:conf/alenex/JabrayilovM19} and shows that the partial-ordering model outperforms the assignment model; it has better LP relaxation values and solves more instances. 
So our new experimental study concerns only the \hstp and \hmst.
\ifshort
\else
To speed up algorithms for these problems, computational studies in the literature use problem-specific strengthening constraints and utilize reduction techniques that can eliminate up to $69\%$ of the edges~\cite{Akg_n_2011}.
Since our experiments aim only to compare the two models, we did not use reduction techniques or strengthening constraints.
\fi

To solve the models, we used the Gurobi 6.5.1 single-threadedly on the Intel Xeon E5-2640 2.60GHz system running Ubuntu 18.04.
We performed our tests on the benchmark instances used in the literature (\eg~\cite{Akg_n_2011,GOUVEIA2001539}) for the \hmst, namely Euclidean (TC, TE) and random (TR) complete graphs with 21, 41, 61, 81 nodes and up to 3240 edges.
For simplicity, they are referred to as 20, 40, 60, 80. 
The last node in set $V=\{1,\dots,|V|\}$ is used as the root, \ie,  $r=|V|$.
Based on each of these \hmst instances, we also created an \hstp instance with %
$R \ne V$, namely 
$R=\{1,\dots,\lfloor \frac{|V|}{2} \rfloor \} \cup \{r\}$, 
\ie, the first 
$\lfloor \frac{|V|}{2} \rfloor$ 
nodes and the root $r=|V|$ are terminals.
We tested nine hop parameters $H=2,\dots,10$.
Recall that the \hmst instances are \hstp instances with $R=V$, %
so we get for each of the TC, TE, TR graphs and each $H=2,\dots,10$, two \hstp instances, one with %
$R=V$ and one with $R \ne V$;
this leads to a total of 216 instances.
\def\OptHat{\hat{O}}
\def\OptHat{\widehat{Opt}}
\def\OptHat{\overline{Opt}}
The evaluations 
\ifshort
(for detailed numbers, see the full version \cite{jabrayilov2020hopconstrained}) 
\else
(for detailed numbers, see the appendix) 
\fi
of LP relaxations show that the inequality 
$\val_{\text{\ref{Phstp}}} \ge \val_{\text{\ref{Ahstp}}}$ 
holds for all 216 instances.
\ifshort
Figure \ref{plotNumberStrong}  shows for each $H$ the number of instances for which the inequality is strict.
\else
Figure \ref{plotNumberStrong}  shows for each $H=2,\dots,10$ the number of instances for which the inequality is strict, \ie, the LP relaxation value of (\ref{Phstp}) is greater than that of (\ref{Ahstp}).
\fi
The strict inequality holds mostly for small hop limits,
\eg, for 23 of 24 instances with $H=2$ but only for 2 of 24 instances with $H=10$. %
It is also interesting to compare the LP gaps of the two models, \ie, the gap between optimal integer value $Opt$ and the LP relaxation value.
If $Opt$ could not be computed, we used the value $\OptHat$ of the best integer solution found by the two models.
Figure~\ref{plotMeangap} shows the average LP gaps of the two models over 216 instances depending on $(|V|,H)$, where the LP gap of model $M$ is given by $\frac{\OptHat-\val_M}{\OptHat}$.
\ifshort
The figure indicates that the gaps become larger as $|V|$ increases and $H$ decreases.
\else
The figure indicates that the gaps become larger as the number $|V|$ of nodes increases and the hop limit $H$ decreases.
\fi
Moreover, (\ref{Phstp}) has smaller gaps than (\ref{Ahstp}) for small hop limits, where the small hop limits significantly affect the objective value.

\ifshort
\else
\newdimen\mydim	    
\mydim=0.49\textwidth
\def\scale{0.60}
\begin{figure*}[bt]
  \captionsetup[subfigure]{justification=centering}
  \centering
  \begin{subfigure}[b] {\mydim}
  \begin{tikzpicture}
    \path ( 0,0)    node {\includegraphics[scale=\scale]{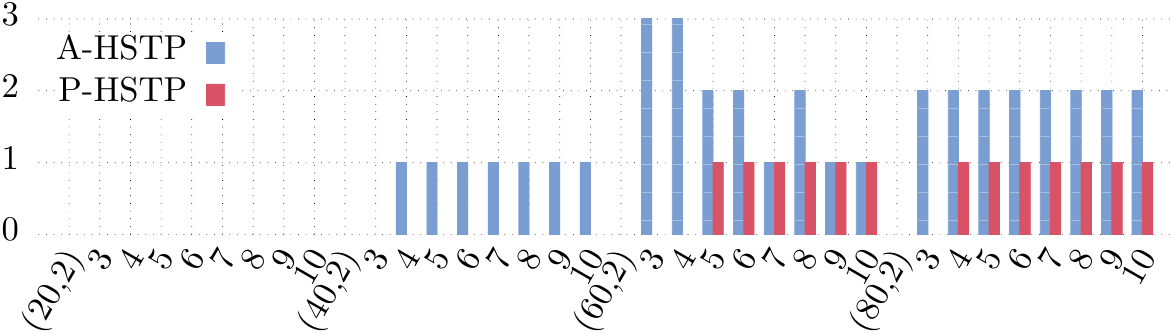}};
  \end{tikzpicture}
	\vspace{-7pt}
	\subcaption{}
	\label{plotUnsolvedHMST}
  \end{subfigure}
  \begin{subfigure}[b] {\mydim}
  \begin{tikzpicture}
    \path ( 0,0)    node {\includegraphics[scale=\scale]{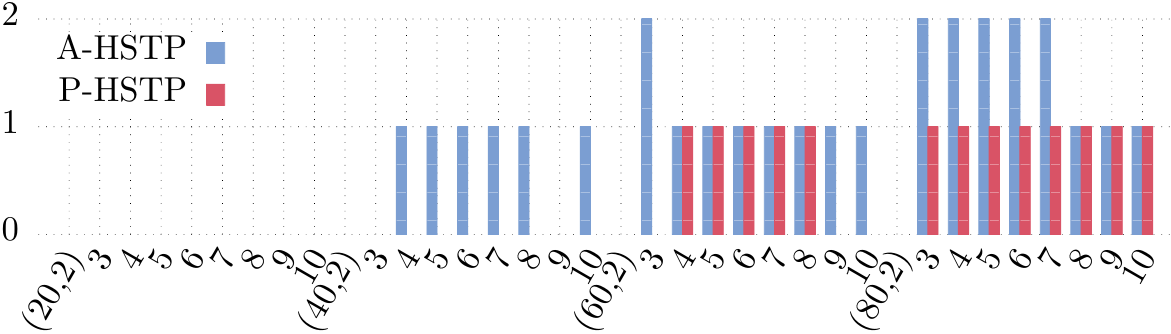}};
  \end{tikzpicture}
	\vspace{-7pt}
	\subcaption{}
	\label{plotUnsolvedHSTP}
  \end{subfigure}
  \caption{Number of the unsolved (\subref{plotUnsolvedHMST}) \hmst, (\subref{plotUnsolvedHSTP}) \hstp instances depending on $(|V|,H)$
   }
  \label{plotUnsolved}
\end{figure*}
\fi

\ifshort
\def\scale{0.5}
\begin{figure*}[bt]
  \captionsetup[subfigure]{justification=centering}
  \begin{subfigure}[b] {0.39\textwidth}
  \centering
  \begin{tikzpicture}
    \path (0, 0.0)    node {\includegraphics[scale=\scale]{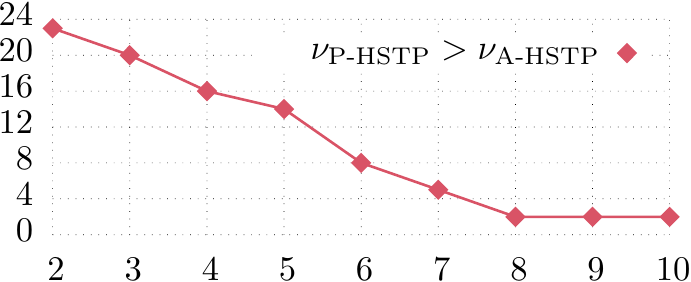}};
    \path ( 0,-1.08)    node {};
  \end{tikzpicture}
	\vspace{-12pt}
	\subcaption{}
	\label{plotNumberStrong}
  \end{subfigure}
  \begin{subfigure}[b] {0.59\textwidth}
  \begin{tikzpicture}
    \path ( 0,0)    node {\includegraphics[scale=\scale]{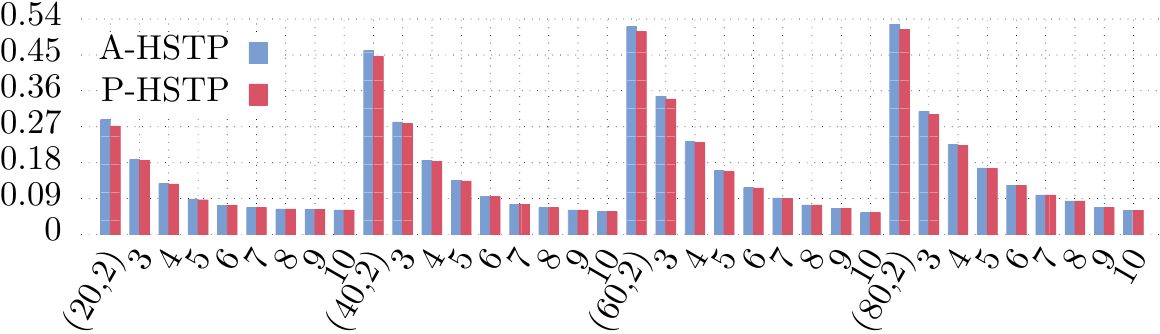}};
  \end{tikzpicture}
	\vspace{-9pt}
	\subcaption{}
	\label{plotMeangap}
  \end{subfigure}
  \vspace{-9pt}
  \caption{
    (\subref{plotNumberStrong}) Number of instances 
    with strict inequality $\val_{\text{\ref{Phstp}}} > \val_{\text{\ref{Ahstp}}}$ 
    depending on $H$;
    (\subref{plotMeangap}) Average LP gaps %
    over 216 instances
    depending on $(|V|,H)$
  }
  \label{plotGap}
\end{figure*}
\else
\fi

\ifshort
\newdimen\mydim	    
\newdimen\mydimII
\mydim=0.37\textwidth
\mydimII=0.24\textwidth
\def\scale{0.40}
\begin{figure*}[bt]
  \captionsetup[subfigure]{justification=centering}
  \centering
  \begin{subfigure}[b] {\mydim}
  \begin{tikzpicture}
    \path ( 0,0)    node {\includegraphics[scale=\scale]{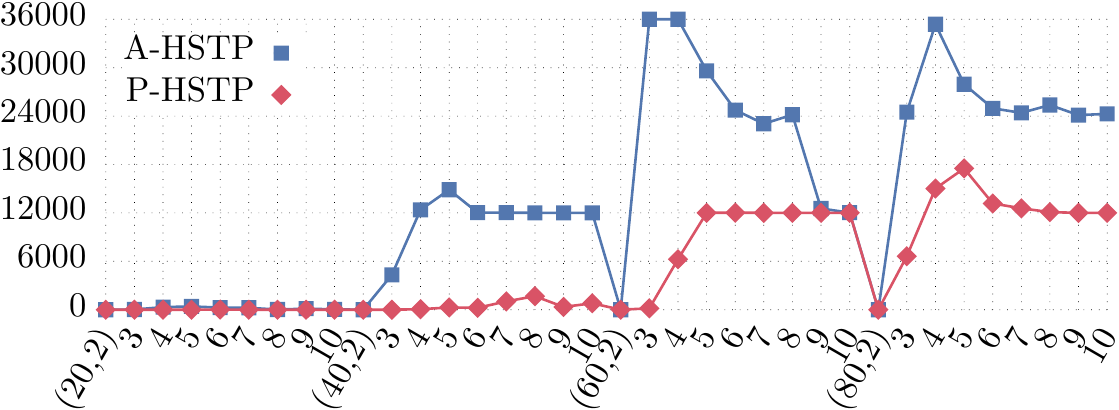}};
  \end{tikzpicture}
	\vspace{-19pt}
	\subcaption{}
	\label{plotMeanTimeHMST}
  \end{subfigure}
  \begin{subfigure}[b] {\mydim}
  \begin{tikzpicture}
    \path ( 0,0)    node {\includegraphics[scale=\scale]{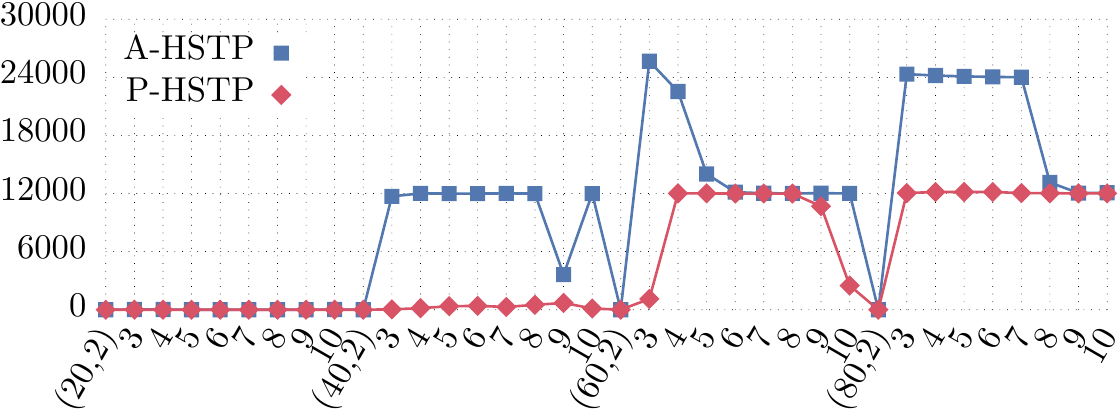}};
  \end{tikzpicture}
	\vspace{-19pt}
	\subcaption{}
	\label{plotMeanTimeHSTP}
  \end{subfigure}
  \begin{subfigure}[b] {\mydimII}
  \begin{tikzpicture}
    \path ( 0,0)    node {\includegraphics[scale=\scale]{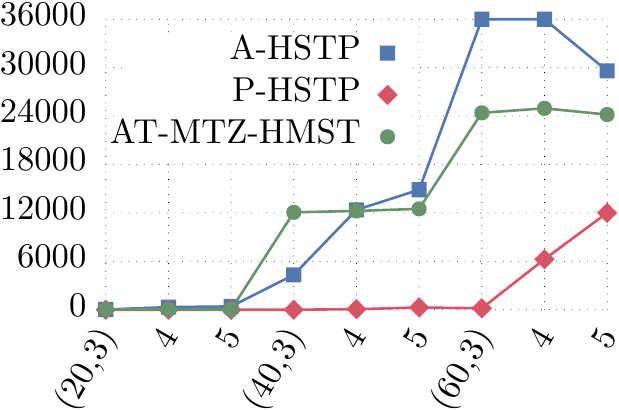}};
  \end{tikzpicture}
	\vspace{-8pt}
	\subcaption{}
	\label{plotPopAssMTZ}
  \end{subfigure}
  \vspace{-17pt}
  \caption{Average run times of (\ref{Ahstp}) and (\ref{Phstp}) for the (\subref{plotMeanTimeHMST}) \hmst, (\subref{plotMeanTimeHSTP}) \hstp
  instances; %
   (\subref{plotPopAssMTZ})~Comparison of the two models with Akgün-Tansel model~\cite{Akg_n_2011}
   }
  \label{plotRT}
\end{figure*}
\else
\newdimen\mydim	    
\newdimen\mydimII
\mydim=0.37\textwidth
\mydimII=0.24\textwidth
\def\scale{0.50}
\begin{figure*}[bt]
  \captionsetup[subfigure]{justification=centering}
  \centering
  \begin{subfigure}[b] {\mydim}
  \begin{tikzpicture}
    \path ( 0,0)    node {\includegraphics[scale=\scale]{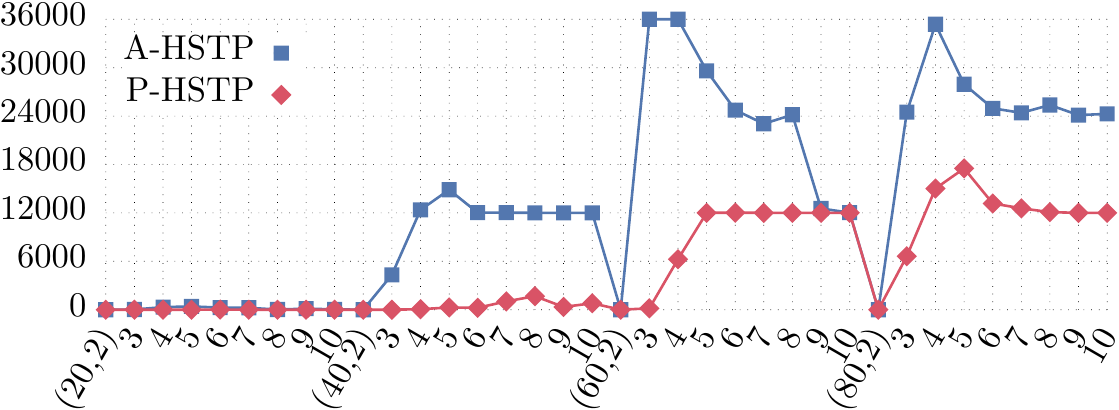}};
  \end{tikzpicture}
	\vspace{-19pt}
	\subcaption{}
	\label{plotMeanTimeHMST}
  \end{subfigure}
  \begin{subfigure}[b] {\mydim}
  \begin{tikzpicture}
    \path ( 0,0)    node {\includegraphics[scale=\scale]{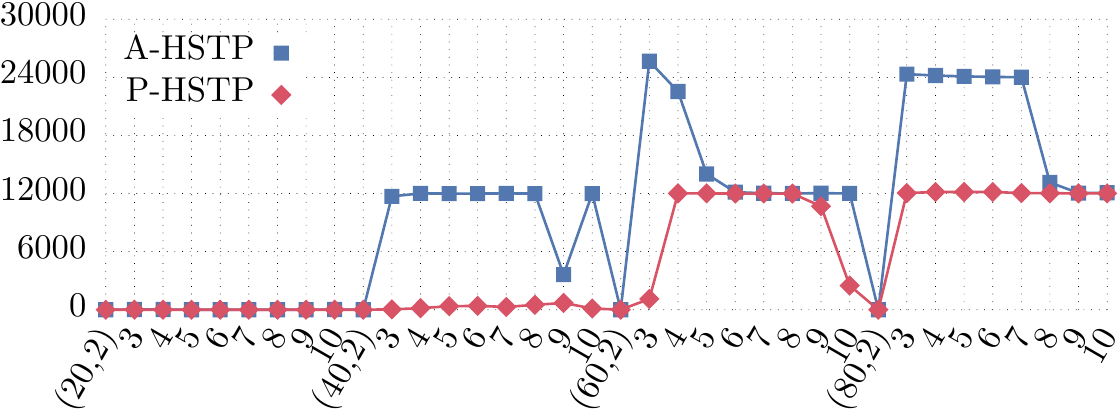}};
  \end{tikzpicture}
	\vspace{-19pt}
	\subcaption{}
	\label{plotMeanTimeHSTP}
  \end{subfigure}
  \begin{subfigure}[b] {\mydimII}
  \begin{tikzpicture}
    \path ( 0,0)    node {\includegraphics[scale=\scale]{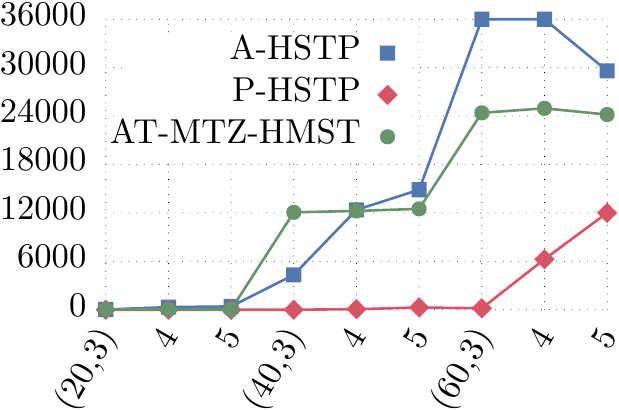}};
  \end{tikzpicture}
	\vspace{-6pt}
	\subcaption{}
	\label{plotPopAssMTZ}
  \end{subfigure}
  \caption{Average run times of (\ref{Ahstp}) and (\ref{Phstp}) for the (\subref{plotMeanTimeHMST}) \hmst, (\subref{plotMeanTimeHSTP}) \hstp
  instances depending on $(|V|,H)$;
   (\subref{plotPopAssMTZ})~Comparison of the two models with Akgün-Tansel model~\cite{Akg_n_2011}
   }
  \label{plotRT}
\end{figure*}
\fi

To compute the optimal integer values, we ran (\ref{Ahstp}) and (\ref{Phstp}) with a time limit of 10 hours.
\ifshort
\else
Figure \ref{plotUnsolved} shows the number of unsolved \hmst (Fig. \ref{plotUnsolvedHMST}) and  \hstp (Fig. \ref{plotUnsolvedHSTP}) instances,
while 
\fi
Figure \ref{plotRT} shows the average run times in seconds for the \hmst (Fig. \ref{plotMeanTimeHMST}) and \hstp (Fig. \ref{plotMeanTimeHSTP}) instances in dependency of $(|V|,H)$.
Model (\ref{Ahstp}) solves 100 instances within a few seconds and needs several minutes for 37 instances and hours for 13 instances. 
It misses solving the remaining 66 instances.
Model (\ref{Phstp}) outperforms (\ref{Ahstp}) and solves 156 instances within a few seconds and needs some minutes for 27 instances and hours for 7 instances. 
The model solves 40 instances more than (\ref{Ahstp}).
We also compared the two models with another node-oriented state-of-the-art model suggested by Akgün-Tansel~\cite{Akg_n_2011}, based on MTZ constraints.
Computational results, taken from~\cite{Akg_n_2011}, involve nine TC, TE, TR  graphs with up to 60 nodes, 1830 edges, and hop limits 3--5 for the \hmst.
The comparison of average run times of the three models (Fig. \ref{plotPopAssMTZ}), depending on $(|V|,H)$, shows that the partial-ordering model outperforms the other two.

\section{Conclusion} 
\label{sec:conclusion}

\ifshort
In this work, we showed that the partial-ordering model is strictly stronger than the assignment model for the problems \hstp, \hmst, and \stprbh.
\else
In this work, we provided polyhedral results for two node-originated models, called assignment and partial-ordering based models for the hop-constrained Steiner tree problems \hstp, \hmst, and \stprbh.
We showed that the partial-ordering model is strictly stronger than the assignment model for these problems.
\fi
\ifshort
Moreover, Theorem \ref{theorem:hstp:fme} allows us to summarize the polyhedral results as follows: 
For the problems whose solution is a hop-constrained tree, a rooted tree that has bounded depth, the partial-ordering model is stronger than the assignment model if both models have the same objective function and constraints (\ie, depending only on common $x$-variables), except the hop-constrained tree constraints, namely $(x,y) \in \passHT$~(\ref{pass1}) and  $(x,l,g) \in \ppopHT$~(\ref{ppop}).
\else
Corollary \ref{corollaryHT} summarizes these polyhedral results for the problems whose solution is a hop-constrained tree, \ie, a rooted tree with bounded depth.
Moreover, in Theorem \ref{theorem:pop:mtz}, we showed that the partial-ordering based model implies an exponential-sized set of hop-constrained path constraints, which does not hold for the assignment model.
\fi
Furthermore, the computational results in the literature and this work show for the problems \hstp, \hmst, and \stprbh that the partial-ordering based model outperforms the assignment model in practice, too; it has a smaller LP gap and solves more instances.

\bibliographystyle{abbrv}
\bibliography{hstp.bib}

\ifshort
\else

\appendix
\section{Computational results of the assignment and partial-ordering based models} %
\label{appendix:table}

Table \ref{table} shows the results of the assignment model (\ref{Ahstp}) and partial-ordering based model (\ref{Phstp}) for the \hmst and \hstp (with $R \ne V$) instances.
  The first two columns show the evaluated graphs and hop limits.
  The next six columns, namely columns 3--8, show the results for the 108 \hmst instances, while the last six columns show the relevant 
  results for the 108 \hstp (with $R \ne V$) instances. 
  Columns 3-5 show the following results omitted by (\ref{Ahstp}): %
  Column 3 contains corresponding LP relaxation values.
  If the optimal integer value of an instance is found, then Column 4 contains this value; otherwise, it includes an interval $[lb-ub]$, where $lb$ and $ub$ are the lower and upper bounds omitted by IP solver within the time limits of 10 hours.          
  Column 5 includes the times required to find the optimal integer values.        
  Similarly, columns 6-8 show the relevant results omitted by (\ref{Phstp}).

\begin{table}[b] %
  \centering
  \scalebox{0.52}{
    \small
  \setlength{\tabcolsep}{10.0pt} 
  \begin{tabular}{ll | ccc c ccc c || ccc c ccc c}
  \hline
  &&  \multicolumn{7}{c}{\textbf{\hmst} instances ($R=V$)} 
  &&  \multicolumn{7}{c}{\textbf{\hstp} instances  ($R\ne V$)}\\ 
  \cline{4-8}	\cline{12-17} 
  && \multicolumn{3}{c}{(\ref{Ahstp})} && \multicolumn{3}{c}{(\ref{Phstp})} 
  && \multicolumn{3}{c}{(\ref{Ahstp})} && \multicolumn{3}{c}{(\ref{Phstp})} 
  \\
  \cline{3-5}	\cline{7-9}	\cline{12-14}	\cline{16-18}

  Graph &H	&LP   &IP  &Time[s]	    &	    &LP	 &IP	&Time[s]    & &
		&LP   &IP  &Time[s]	    &	    &LP	 &IP	&Time[s]    \\
  \hline
    %TC20 &2    &311.33   &384          &0      &&318.0    &384          &0      &207.36    &249         &0      &&209.78  &249         &0    \\
TC20  &2    &311.33   &384          &0      &&318.0    &384          &0      &&&207.36  &249         &0      &&209.78  &249          &0      \\
      &3    &305.5    &340          &2      &&307.0    &340          &0      &&&207.06  &227         &0      &&207.19  &227          &0      \\
      &4    &302.75   &318          &0      &&304.0    &318          &0      &&&207.0   &219         &0      &&207.0   &219          &0      \\
      &5    &302.0    &312          &0      &&302.86   &312          &0      &&&207.0   &219         &0      &&207.0   &219          &0      \\
      &6    &302.0    &302          &0      &&302.0    &302          &0      &&&207.0   &219         &0      &&207.0   &219          &0      \\
      &7    &302.0    &302          &0      &&302.0    &302          &0      &&&207.0   &219         &0      &&207.0   &219          &0      \\
      &8    &302.0    &302          &0      &&302.0    &302          &0      &&&207.0   &219         &1      &&207.0   &219          &0      \\
      &9    &302.0    &302          &0      &&302.0    &302          &0      &&&207.0   &219         &0      &&207.0   &219          &0      \\
      &10   &302.0    &302          &0      &&302.0    &302          &0      &&&207.0   &219         &1      &&207.0   &219          &0      \\
TC40  &2    &472.67   &747          &0      &&479.0    &747          &0      &&&284.0   &431         &0      &&286.74  &431          &0      \\
      &3    &472.0    &609          &733    &&472.4    &609          &2      &&&283.0   &345         &93     &&283.0   &345          &1      \\
      &4    &472.0    &548          &1005   &&472.0    &548          &5      &&&281.5   &317         &5      &&281.5   &317          &1      \\
      &5    &472.0    &522          &8563   &&472.0    &522          &6      &&&281.5   &313         &6      &&281.5   &313          &1      \\
      &6    &472.0    &498          &89     &&472.0    &498          &4      &&&281.5   &310         &9      &&281.5   &310          &2      \\
      &7    &472.0    &490          &78     &&472.0    &490          &2      &&&281.5   &306         &17     &&281.5   &306          &2      \\
      &8    &472.0    &488          &33     &&472.0    &488          &2      &&&281.5   &302         &11     &&281.5   &302          &1      \\
      &9    &472.0    &484          &19     &&472.0    &484          &1      &&&281.5   &300         &6      &&281.5   &300          &1      \\
      &10   &472.0    &482          &34     &&472.0    &482          &2      &&&281.5   &300         &10     &&281.5   &300          &1      \\
TC60  &2    &665.33   &1083         &1      &&678.27   &1083         &0      &&&417.06  &635         &1      &&426.48  &635          &0      \\
      &3    &658.67   &[821-866]    &36000  &&661.87   &866          &13     &&&408.99  &528         &4983   &&411.51  &528          &26     \\
      &4    &658.4    &[733-785]    &36000  &&659.0    &781          &71     &&&407.99  &487         &19940  &&408.73  &487          &28     \\
      &5    &658.27   &[711-734]    &36000  &&658.8    &734          &40     &&&407.67  &458         &1235   &&407.93  &458          &49     \\
      &6    &658.19   &[692-718]    &36000  &&658.67   &712          &56     &&&407.13  &440         &195    &&407.13  &440          &6      \\
      &7    &658.0    &692          &30317  &&658.31   &692          &15     &&&406.73  &432         &36     &&406.73  &432          &5      \\
      &8    &658.0    &[670-682]    &36000  &&658.0    &682          &9      &&&406.73  &432         &19     &&406.73  &432          &6      \\
      &9    &658.0    &676          &1209   &&658.0    &676          &5      &&&406.73  &432         &110    &&406.73  &432          &13     \\
      &10   &658.0    &672          &44     &&658.0    &672          &4      &&&406.73  &432         &43     &&406.73  &432          &7      \\
TC80  &2    &820.08   &1305         &1      &&832.67   &1305         &0      &&&492.84  &785         &3      &&498.89  &785          &0      \\
      &3    &816.0    &[958-1083]   &36000  &&816.57   &1072         &312    &&&488.62  &[593-646]   &36000  &&490.8   &643          &161    \\
      &4    &816.0    &[856-990]    &36000  &&816.0    &981          &9055   &&&487.34  &[547-605]   &36000  &&488.59  &587          &451    \\
      &5    &816.0    &[842-926]    &36000  &&816.0    &922          &16566  &&&487.05  &[536-564]   &36000  &&487.54  &560          &487    \\
      &6    &816.0    &[828-912]    &36000  &&816.0    &884          &3493   &&&486.8   &[523-544]   &36000  &&487.0   &544          &455    \\
      &7    &816.0    &[829-872]    &36000  &&816.0    &862          &1685   &&&486.4   &[520-528]   &36000  &&486.5   &528          &139    \\
      &8    &816.0    &[825-846]    &36000  &&816.0    &846          &264    &&&486.0   &518         &3425   &&486.0   &518          &91     \\
      &9    &816.0    &[821-838]    &36000  &&816.0    &838          &40     &&&486.0   &512         &129    &&486.0   &512          &30     \\
      &10   &816.0    &[820-834]    &36000  &&816.0    &834          &40     &&&486.0   &510         &201    &&486.0   &510          &18     \\
TE20  &2    &287.14   &561          &0      &&292.83   &561          &0      &&&197.0   &362         &0      &&199.3   &362          &0      \\
      &3    &284.0    &449          &71     &&284.0    &449          &1      &&&196.0   &290         &43     &&196.0   &290          &1      \\
      &4    &284.0    &385          &972    &&284.0    &385          &7      &&&196.0   &268         &35     &&196.0   &268          &3      \\
      &5    &284.0    &366          &1250   &&284.0    &366          &16     &&&196.0   &247         &15     &&196.0   &247          &1      \\
      &6    &284.0    &355          &780    &&284.0    &355          &36     &&&196.0   &239         &14     &&196.0   &239          &2      \\
      &7    &284.0    &344          &714    &&284.0    &344          &25     &&&196.0   &236         &6      &&196.0   &236          &2      \\
      &8    &284.0    &336          &80     &&284.0    &336          &9      &&&196.0   &236         &9      &&196.0   &236          &2      \\
      &9    &284.0    &332          &391    &&284.0    &332          &11     &&&196.0   &236         &13     &&196.0   &236          &3      \\
      &10   &284.0    &330          &160    &&284.0    &330          &7      &&&196.0   &236         &40     &&196.0   &236          &4      \\
TE40  &2    &488.47   &915          &0      &&498.4    &915          &0      &&&282.14  &561         &0      &&288.17  &561          &0      \\
      &3    &481.13   &708          &12070  &&484.67   &708          &25     &&&279.0   &449         &35030  &&279.0   &449          &113    \\
      &4    &480.64   &[542-630]    &36000  &&481.83   &627          &176    &&&279.0   &[364-385]   &36000  &&279.0   &385          &500    \\
      &5    &480.29   &[533-595]    &36000  &&481.17   &590          &859    &&&279.0   &[338-366]   &36000  &&279.0   &366          &1098   \\
      &6    &480.0    &[511-572]    &36000  &&480.44   &565          &675    &&&279.0   &[335-358]   &36000  &&279.0   &355          &1164   \\
      &7    &480.0    &[506-550]    &36000  &&480.0    &544          &3095   &&&279.0   &[325-345]   &36000  &&279.0   &344          &860    \\
      &8    &480.0    &[503-536]    &36000  &&480.0    &536          &5060   &&&279.0   &[334-336]   &36000  &&279.0   &336          &1490   \\
      &9    &480.0    &[503-529]    &36000  &&480.0    &528          &1031   &&&279.0   &332         &10858  &&279.0   &332          &2116   \\
      &10   &480.0    &[502-520]    &36000  &&480.0    &520          &2376   &&&279.0   &[325-330]   &36000  &&279.0   &330          &391    \\
TE60  &2    &970.4    &2108         &1      &&987.67   &2108         &0      &&&635.93  &1348        &1      &&641.58  &1348         &1      \\
      &3    &955.39   &[1224-1550]  &36000  &&961.81   &1525         &483    &&&633.9   &[782-1011]  &36000  &&635.77  &1004         &3279   \\
      &4    &951.22   &[1037-1351]  &36000  &&953.67   &1336         &18680  &&&633.3   &[735-907]   &36000  &&634.65  &[846-901]    &36000  \\
      &5    &950.0    &[1010-1233]  &36000  &&951.08   &[1202-1226]  &36000  &&&632.63  &[720-842]   &36000  &&633.4   &[785-835]    &36000  \\
      &6    &950.0    &[993-1159]   &36000  &&950.0    &[1096-1154]  &36000  &&&632.4   &[710-816]   &36000  &&632.4   &[751-802]    &36000  \\
      &7    &950.0    &[979-1126]   &36000  &&950.0    &[1046-1107]  &36000  &&&632.4   &[700-801]   &36000  &&632.4   &[738-784]    &36000  \\
      &8    &950.0    &[978-1105]   &36000  &&950.0    &[1026-1072]  &36000  &&&632.4   &[700-774]   &36000  &&632.4   &[738-760]    &36000  \\
      &9    &950.0    &[978-1076]   &36000  &&950.0    &[1012-1053]  &36000  &&&632.4   &[706-748]   &36000  &&632.4   &745          &32059  \\
      &10   &950.0    &[981-1101]   &36000  &&950.0    &[1008-1036]  &36000  &&&632.4   &[703-731]   &36000  &&632.4   &731          &7471   \\
TE80  &2    &1090.6   &2547         &1      &&1104.0   &2547         &1      &&&717.56  &1600        &3      &&726.26  &1600         &2      \\
      &3    &1077.53  &[1327-1868]  &36000  &&1086.39  &1806         &19575  &&&714.37  &[857-1219]  &36000  &&717.08  &[1084-1172]  &36000  \\
      &4    &1074.23  &[1195-1637]  &36000  &&1079.0   &[1476-1573]  &36000  &&&713.0   &[793-1079]  &36000  &&714.14  &[885-1051]   &36000  \\
      &5    &1072.53  &[1144-1487]  &36000  &&1075.3   &[1322-1455]  &36000  &&&713.0   &[768-1011]  &36000  &&713.0   &[843-964]    &36000  \\
      &6    &1072.0   &[1134-1402]  &36000  &&1072.73  &[1241-1348]  &36000  &&&713.0   &[762-979]   &36000  &&713.0   &[807-932]    &36000  \\
      &7    &1072.0   &[1124-1335]  &36000  &&1072.0   &[1201-1310]  &36000  &&&713.0   &[766-922]   &36000  &&713.0   &[796-897]    &36000  \\
      &8    &1072.0   &[1118-1345]  &36000  &&1072.0   &[1171-1274]  &36000  &&&713.0   &[771-916]   &36000  &&713.0   &[783-873]    &36000  \\
      &9    &1072.0   &[1121-1334]  &36000  &&1072.0   &[1159-1242]  &36000  &&&713.0   &[761-866]   &36000  &&713.0   &[786-835]    &36000  \\
      &10   &1072.0   &[1109-1314]  &36000  &&1072.0   &[1150-1226]  &36000  &&&713.0   &[770-876]   &36000  &&713.0   &[775-824]    &36000  \\
TR20  &2    &145.5    &227          &0      &&157.33   &227          &0      &&&103.0   &110         &0      &&103.0   &110          &0      \\
      &3    &138.67   &168          &0      &&140.0    &168          &0      &&&102.58  &110         &0      &&102.67  &110          &0      \\
      &4    &137.0    &146          &0      &&138.14   &146          &0      &&&102.5   &110         &0      &&102.5   &110          &0      \\
      &5    &137.0    &137          &0      &&137.0    &137          &0      &&&102.0   &102         &0      &&102.0   &102          &0      \\
      &6    &137.0    &137          &0      &&137.0    &137          &0      &&&102.0   &102         &0      &&102.0   &102          &0      \\
      &7    &137.0    &137          &0      &&137.0    &137          &0      &&&102.0   &102         &0      &&102.0   &102          &0      \\
      &8    &137.0    &137          &0      &&137.0    &137          &0      &&&102.0   &102         &0      &&102.0   &102          &0      \\
      &9    &137.0    &137          &0      &&137.0    &137          &0      &&&102.0   &102         &0      &&102.0   &102          &0      \\
      &10   &137.0    &137          &0      &&137.0    &137          &0      &&&102.0   &102         &0      &&102.0   &102          &0      \\
TR40  &2    &132.85   &296          &0      &&139.42   &296          &0      &&&71.5    &155         &0      &&75.93   &155          &0      \\
      &3    &127.18   &176          &186    &&129.0    &176          &0      &&&68.62   &99          &15     &&69.42   &99           &1      \\
      &4    &126.33   &149          &100    &&127.06   &149          &0      &&&68.2    &85          &17     &&68.58   &85           &1      \\
      &5    &126.12   &139          &38     &&126.5    &139          &1      &&&68.0    &75          &1      &&68.2    &75           &0      \\
      &6    &126.0    &130          &2      &&126.0    &130          &0      &&&68.0    &70          &1      &&68.0    &70           &0      \\
      &7    &126.0    &127          &1      &&126.0    &127          &0      &&&68.0    &69          &1      &&68.0    &69           &0      \\
      &8    &126.0    &127          &2      &&126.0    &127          &0      &&&68.0    &69          &1      &&68.0    &69           &0      \\
      &9    &126.0    &127          &2      &&126.0    &127          &1      &&&68.0    &69          &1      &&68.0    &69           &0      \\
      &10   &126.0    &127          &2      &&126.0    &127          &0      &&&68.0    &69          &1      &&68.0    &69           &0      \\
TR60  &2    &145.96   &452          &1      &&155.25   &452          &0      &&&93.54   &268         &1      &&98.5    &268          &0      \\
      &3    &139.13   &[228-256]    &36000  &&142.17   &256          &7      &&&90.7    &[144-153]   &36000  &&92.6    &153          &9      \\
      &4    &137.65   &[166-193]    &36000  &&138.95   &187          &4      &&&90.25   &116         &11708  &&90.78   &116          &14     \\
      &5    &137.38   &158          &16815  &&138.0    &158          &3      &&&90.11   &105         &4823   &&90.33   &105          &8      \\
      &6    &137.2    &149          &2245   &&137.83   &149          &4      &&&90.07   &98          &222    &&90.25   &98           &7      \\
      &7    &137.08   &147          &2837   &&137.45   &147          &2      &&&90.0    &93          &37     &&90.08   &93           &1      \\
      &8    &137.05   &143          &509    &&137.31   &143          &3      &&&90.0    &92          &29     &&90.0    &92           &1      \\
      &9    &137.03   &143          &410    &&137.24   &143          &3      &&&90.0    &91          &13     &&90.0    &91           &1      \\
      &10   &137.0    &139          &10     &&137.11   &139          &2      &&&90.0    &91          &21     &&90.0    &91           &1      \\
TR80  &2    &155.3    &464          &1      &&166.77   &464          &0      &&&99.62   &266         &3      &&105.36  &266          &1      \\
      &3    &148.95   &208          &1452   &&151.36   &208          &1      &&&94.51   &132         &1016   &&97.33   &132          &2      \\
      &4    &147.54   &180          &34111  &&148.21   &180          &4      &&&93.67   &116         &585    &&94.52   &116          &12     \\
      &5    &147.07   &164          &11884  &&147.32   &164          &2      &&&93.22   &106         &315    &&93.5    &106          &4      \\
      &6    &147.0    &157          &2913   &&147.07   &157          &8      &&&93.2    &98          &118    &&93.33   &98           &4      \\
      &7    &147.0    &153          &1190   &&147.0    &153          &8      &&&93.17   &96          &9      &&93.25   &96           &2      \\
      &8    &147.0    &153          &4102   &&147.0    &153          &8      &&&93.12   &95          &33     &&93.2    &95           &2      \\
      &9    &147.0    &150          &380    &&147.0    &150          &9      &&&93.12   &95          &42     &&93.17   &95           &2      \\
      &10   &147.0    &150          &810    &&147.0    &150          &7      &&&93.1    &94          &65     &&93.14   &94           &3      \\

  \hline
  \end{tabular}
  }
  \caption{Results of the assignment and partial-ordering based models for the 216 \hmst and \hstp instances.}
  \label{table}
\end{table}
\fi

\end{document}